\documentclass[draft,12pt,reqno,a4paper]{amsart}

\usepackage[margin=3cm,footskip=1cm]{geometry}

\usepackage{amssymb} 
\usepackage{amsmath} 
\usepackage{amsfonts}
\usepackage{amsthm} 
\usepackage{mathtools}
\usepackage{color}

\usepackage{enumerate} 
\usepackage[latin1]{inputenc}
\usepackage[shortlabels]{enumitem}
\usepackage{color}
\usepackage{graphicx} 
\usepackage{verbatim}

\DeclareMathOperator*{\slim}{s--lim}

\DeclareMathOperator*{\vGlim}{{\Sigma}--lim}
\DeclareMathOperator*{\vSigmalim}{{\Sigma}--lim}

\DeclarePairedDelimiter\set{\{}{\}}

\DeclarePairedDelimiter\abs\lvert\rvert
\DeclarePairedDelimiter\norm\lVert\rVert

\newcommand{\supp}{\operatorname{supp}}

\newcommand{\Div}{{\operatorname{div}}}

\newcommand{\N}{{\mathbb{N}}} 
\newcommand{\R}{{\mathbb{R}}} 
\newcommand{\C}{{\mathbb{C}}}

 \renewcommand{\c}{{\rm c}}
\newcommand{\e}{{\rm e}} 
 \renewcommand{\i}{{\rm i}}
\renewcommand{\d}{{\rm d}}

\renewcommand{\Re}{{\rm Re}\,}

\DeclarePairedDelimiter\inp\langle\rangle

\newcommand\parb[2][]{#1 \big ( #2#1\big )} \newcommand\parbb[2][]{#1
  \Big ( #2#1\Big )}

\newcommand{\mand}{\text{ and }}

 \newcommand{\vB}{{\mathcal B}}
 
\newcommand{\vE}{{\mathcal E}} \newcommand{\vF}{{\mathcal F}}
 
 \newcommand{\vH}{{\mathcal H}}
 \newcommand{\vL}{{\mathcal L}}
 
\newcommand{\vO}{{\mathcal O}} 
 \newcommand{\vR}{{\mathcal R}}
 
 \newcommand{\vX}{{\mathcal X}}

\newcommand\Step[1]{
  \par\bigskip
  \noindent
  \textit{#1}.\enspace
}

\newcommand{\intR}{{-\!\!\!\!\!\!\int_{\rho}\,}}

\newcommand{\cas}{{\rm the Cauchy-Schwarz inequality }}
\newcommand{\caS}{{\rm the Cauchy-Schwarz inequality}}

\theoremstyle{plain}
\newtheorem{thm}{Theorem}[section]
\newtheorem{proposition}[thm]{Proposition}
\newtheorem{lemma}[thm]{Lemma} \newtheorem{corollary}[thm]{Corollary}
\theoremstyle{definition}

 \newtheorem{cond}[thm]{Condition}
 \newtheorem{remark}[thm]{Remark}

 \newtheorem*{remarks*}{Remarks}
\newtheorem*{remark*}{Remark}

\numberwithin{equation}{section}

\title {Stationary scattering theory for  $1$-body Stark  operators, II}

\thanks{
K.I. is supported by JSPS KAKENHI grant no.\ 17K05325.\,
E.S. is  supported by the Research Institute for Mathematical Sciences, a Joint
Usage/Research Center located in Kyoto University, and by DFF grant
no.\ 4181-00042. K.I. and E.S. are  supported by the Swedish Research
Council  grant no. 2016-06596  (residencing at Institut Mittag-Leffler in Djursholm, Sweden, during the Spring semester of 2019).
} 

\author{K. Ito}
\address[K. Ito]{Graduate School of Mathematical Sciences, The University of Tokyo\\
3-8-1 Komaba, Meguro-ku, Tokyo 153-8914, Japan}
\email{ito@ms.u-tokyo.ac.jp}

\author{E. Skibsted} \address[E. Skibsted]{Institut for Matematiske Fag\\
Aarhus Universitet\\ Ny Munkegade 8000 Aarhus C, Denmark}
\email{skibsted@math.au.dk}

\date{\today}
\begin{document}
\begin{abstract}
We study and develop  the stationary scattering theory for a class of one-body
Stark Hamiltonians with short-range potentials, including the Coulomb
potential, continuing our study in \cite{AIIS1,AIIS2}. The classical
scattering orbits are parabolas parametrized by asymptotic orthogonal
momenta, and the kernel of the (quantum) scattering matrix at a fixed
energy is defined in these momenta. We show that the 
scattering matrix is a classical type pseudodifferential operator and    compute the   leading order
singularities at the
diagonal of its kernel. Our approach can be viewed as an adaption of the method
of Isozaki-Kitada \cite{IK} used for studying the scattering matrix for  one-body Schr\"odinger operators
without an external potential. It is more flexible and more
informative than  the more
standard method used previously  by Kvitsinsky-Kostrykin \cite{KK1}
for computing the leading order
singularities of the kernel of the scattering matrix   
in the case of a constant  external {field} (the Stark case). Our
approach  relies on  Sommerfeld's uniqueness result in Besov
spaces, microlocal analysis  as well as
on 
classical phase space constructions.
\end{abstract}

\allowdisplaybreaks

\maketitle
\tableofcontents

\section{Introduction and results}\label{sec:Introduction}

In this paper we  continue  a study of  the stationary scattering
theory  for a class of one-body
Stark Hamiltonians  with short-range potentials initiated in
\cite{AIIS2} (the conditions on the potentials will be stronger though). While the time-dependent scattering theory is
well-understood \cite{AH,He,Ya1,Ya2,Wh}, the  stationary scattering
theory is in our opinion on a less complete form,  even for 
short-range potentials. There are other papers in the
literature on   one-body
Stark stationary scattering theory, see for example  \cite{KK1, KK2}, however we found it
useful and appealing  for its intrinsic beauty  to give a more
systematic account  done entirely within the framework of stationary
theory. While the time-dependent framework was not discussed at  all in
\cite{AIIS2}, we make in the present continuation of \cite{AIIS2}  a
link to the time-dependent framework, showing that the studied
quantities are the same (in disguise form of course). Having this
settled there is one remaining issue, which is 
 a deeper  study
of the scattering matrix. 

To obtain detailed properties 
of the scattering matrix is the main goal of the present paper. This
goal is roughly the same as the one of  \cite{KK1} with
which we have overlapping results. However our approach is very
different from Kvitsinsky-Kostrykin's. It can be considered as an adaption of the method 
of Isozaki-Kitada \cite{IK} used for studying the scattering matrix for  one-body Schr\"odinger operators
without an external potential. It is more flexible and more
informative than  `the 
standard method'.  While the  paper by Kvitsinsky-Kostrykin can be seen
as an application of the 
standard method 
in the case of a constant  external potential, our scheme is closer to
the one invented by Isozaki-Kitada. In particular it relies on
microlocal analysis   and 
on 
classical phase space constructions. By using  such notions the
accomplishment  of
\cite{IK} is  a `trivialization' of the detailed study of the scattering
matrix. More precisely   Isozaki-Kitada  realized 
the  scattering
matrix in the form of a  pseudodifferential operator PsDO (this is for
short-range potentials, however  recently extended to a class of 
long-range potentials  \cite{Na}) and  extracted   the singularities
of its kernel from their representation.

In the present  paper we  do the same  in the
Stark case with short-range potentials. (Note that the Coulomb
potential is a `short-range' potential in the Stark case). In
particular, and more
precisely,  we show how to isolate the local singularities to be
present only in a term
expressed as  an 
\emph{explicit oscillatory integral}, which in turn can be realized  as the
kernel of a  classical type PsDO. For example this allows us to compute the
leading order local singularities  for the Coulomb
potential, reproducing  a result in \cite{KK1}. 

Let us outline   the relationship of our purely stationary setup to time-dependent
scattering theory. We consider
a $d$-dimensional particle (with $d\geq 2$) subject to a constant nonzero field
pointing in the $x_1$-direction. For simplicity we assume that its strength
as well as the particle mass and charge are all taken to $1$. 
We split the coordinates in $\R^d$ into $x_1$ and the  coordinates for
orthogonal directions, 
decomposing   the configuration space variable as 
\begin{align*}
(x,y)\in\mathbb R\times \mathbb R^{d-1};\quad x=x_1,\ y=(x_2,\dots,x_d).
\end{align*}
Then the classical free Stark Hamiltonian is given by 
\begin{align*}
h_0(x,y,\eta,\zeta)
=\tfrac12(\eta^2+\zeta^2)-x,\quad 
(x,y,\eta,\zeta)\in T^*\mathbb R^d\cong\mathbb R^{2d},
\end{align*}
and hence the associated Hamilton equations are  
\begin{align*}
\dot x=\eta,\quad 
\dot y=\zeta,\quad
\dot\eta=1,\quad
\dot\zeta=0.
\end{align*}
 The  solution with  the  initial data $(x_0,y_0,\eta_0,\zeta_0)\in
 T^*\mathbb R^d$,  defining the free classical flow (say denoted by  $\Theta(t)$),  is given by 
\begin{align}\label{eq:freeClas}
x=\tfrac12t^2+t\eta_0+x_0,\quad
y=t\zeta_0+y_0,\quad
\eta=t+\eta_0,\quad
\zeta=\zeta_0.
\end{align}

In particular  the classical orbits are parabolas of the form
$x=\frac1{2\zeta^2}y^2+O(t)$.  The same asymptotics   holds for
$h=h_0+q$, where $q$ is short-range, for example given as $q=q_1$ in
the following condition which will be imposed throughout this paper.

\begin{cond}\label{cond:one-body-starkPot} The  potential $q$ splits into real-valued functions as 
  $q=q_1+q_2$, where   $q_2$ is compactly supported,
  $q_2(-\Delta+1)^{-1}$ is compact,  $q_1$ is smooth  and for some $\delta\in(0,1/2]$
  \begin{align}\label{eq:cond2}
    \partial^\beta q_1=\vO\parb{r^{-(1/2+\delta+\abs{\beta})}};\quad r=(x^2+y^2)^{1/2}.
  \end{align}
\end{cond}

We introduce for any scattering orbit (with potential $q=q_1$)
the
\emph{asymptotic orthogonal momenta} $\zeta^{\pm} =\lim_{t\to \pm
  \infty} \zeta(t)$. The  orbit is incoming and outgoing along 
parabolas given as sections of the paraboloids
$x=\frac1{2(\zeta^\mp)^2}y^2$, respectively. It is part of the
classical scattering problem to determine the transition from an
incoming asymptotic parabola to an outgoing asymptotic parabola, or
stated somewhat strongerly,  the transition from an incoming momentum  $\zeta^-$ to
an outgoing momentum $\zeta^+$. As we will outline below (with further details
given in Section \ref{subsec:Identificaton of wave operators}) this information
is in quantum mechanics  encoded in the
subject of study 
in  this paper, the
scattering matrix.

Under 
Condition \ref{cond:one-body-starkPot} the  free and the perturbed Stark operators on $\vH:=L^2(\R^d)$ are  given by
$H_0=p^2/2-x$ and   $H=p^2/2-x+q$ (with $p=-\i \nabla$), respectively.
We shall throughout this paper impose Condition
\ref{cond:one-body-starkPot} as well as the unique continuation
principle  \cite[Condition 2.4]{AIIS1}. It is a
well-established fact that  
asymptotic completeness holds, i.e.  that the wave operators
\begin{align*}
  W^{\pm}=\slim_{t\to \pm\infty}\e^{\i tH}\e^{-\i tH_0}
\end{align*} exist and map onto $L^2(\R^d)$. The
{asymptotic orthogonal momenta} read in this case
\begin{align*}
  p_y^{\pm} =\lim_{t\to \pm \infty} \e^{\i tH}p_y \e^{-\i
  tH}=\lim_{t\to \pm \infty} \e^{\i tH}y/t\e^{-\i tH},
\end{align*} where the limits are  taken in the strong resolvent
sense (cf. \cite {Ad}).

In terms of the \emph{stationary wave operators} $\vF^-$ and $\vF^+$ we can
simultaneously diagonalize either $H$ and $p_y^{-} $ or  $H$ and
$p_y^{+} $, respectively. These  wave operators  have several
representations, for example given in terms of asymptotic properties
of the boundary values $\lim_{\epsilon \to 0} (H-\lambda \mp
\i\epsilon)^{-1}$ (taken in an appropriate space). However they are
also   connected to  the time-depending wave operators by the formulas
$\vF_0(W^\pm)^*=\vF^\pm$, where  $\vF_0$ is given in terms of asymptotic properties
of the boundary values $\lim_{\epsilon \to 0} (H_0-\lambda \mp
\i\epsilon)^{-1}$, or alternatively and more useful,  given by a Fourier-Airy
transformation which in turn is defined by an explicit oscillatory
integral. By  joint diagonalization we mean  more precisely the assertions
\begin{align*}
  H=\parb{\vF^\pm}^* M_\lambda \vF^\pm \quad \mand \quad p_y^{\pm} =\parb{\vF^\pm}^*
  \parbb{\int _\R \oplus M_{\zeta}\,\d \lambda}\vF^\pm,
\end{align*}
where $ M_{(\cdot)}$ refers to multiplication in $L^2(\R, \d
\lambda;\Sigma)$ or in $\Sigma:=L^2\parb{\R^{d-1}_\zeta,\d
 \zeta}$,
respectively.

The \emph{scattering
  operator} $S=(W^+)^* W^-$ is represented   as 
\begin{align*}
  \vF_0S\vF_0^{-1}=\int _\R \oplus S(\lambda)\,\d \lambda,
\end{align*} where $ S(\lambda)$ is a unitary operator on $\Sigma$ called the \emph{scattering
  matrix }  at energy $\lambda$. Its 
Schwartz kernel 
 $ S(\lambda)(\zeta,\zeta')$ is defined in
 terms of  variables $\zeta$ and $\zeta'$ which by the above formulas may be interpreted
 as outgoing and incoming asymptotic orthogonal momenta,
 respectively. This   explains the  physical relevance of detailed
 information on  $ S(\lambda)$ and its kernel.

In the main  part of the  paper, Sections \ref{subsubsec:Resolvent bounds}--\ref{subsubsec:better Best results on the scattering
  matrix},  we study somewhat refined 
representations of  $ S(\lambda)$ 
based on Sommerfeld's uniqueness result in Besov
spaces proven in \cite{AIIS1}. We show  mapping properties and
we show that the principal symbol of $T(\lambda):=S(\lambda)-I$ viewed
of as a PsDO  is  given by 
\begin{align*}
  t_{\rm psym}(\zeta,\zeta',y)=t_{\rm psym}(y)= -2\i\int_0^\infty\,
                                                                  \tfrac
                                                                  {q_1(x,-y)}{\sqrt{2x}}\,\d
  x.
  \end{align*} 
 In general a linear operator  $T: C_\c^\infty(\R^{d-1})\to
 C^\infty(\R^{d-1})$ is called a \emph{pseudodifferential operator} on
 $\Sigma$ of  \emph{order} $k\in \R$ if
 there exists a smooth function $t=t(\zeta,\zeta',y)$ such that the
 kernel of $T$ is given by 
\begin{align}\label{eq:orderDef} 
  \begin{split}
   T(\zeta,\zeta') &=(2\pi)^{1-d}\int\, \e^{\i(\zeta-\zeta')\cdot
  y} \,t(\zeta,\zeta',y) \,\d y;\\
&\forall \alpha,\alpha', \beta\in \N_0^{d-1}:\\ \abs{ \partial^{\alpha}_{\zeta}\partial^{\alpha'}_{\zeta'}\partial^\beta_y
  t}&\leq
C_{\alpha,\alpha',\beta}\,\inp{y}^{k-\abs{\beta}}\text{
  for all } y \text{ and
  locally  uniformly in }\zeta,\zeta'. 
  \end{split}
\end{align} Here   $\N_0=\N\cup\set{0}$, and we call $t$ {a \emph{symbol}} of $T$. If $k$ can be taken
arbitrarily, $T$  is a \emph{smoothing operator}. We show
that the above operator 
$T(\lambda)$ has order $-\delta$, like the quantization $T_{\rm psym}$
of the symbol $t_{\rm psym}$ has (see
\eqref{eq:elebnd}), while $T(\lambda)-T_{\rm psym} $ has order $-2\delta$.

Applied to $q=\kappa r^{-1}$ for   $d\geq 3$ the singularity
  structure of the
  kernel of the scattering matrix at the diagonal is 
  \begin{align*}
  S(\lambda)(\zeta,
  \zeta')-\delta(\zeta,
  \zeta')=  \kappa C_d\abs{\zeta-\zeta'}^{3/2-d}+ \vO\parb{\abs{\zeta-\zeta'}^{2-d}},
  \end{align*}
 locally uniformly in
  $\zeta,\zeta'$ and $\lambda$. 
This result conforms with \cite{KK1} for $d= 3$ in which case
$C_d=-\i(2\pi)^{-1/2}$  and the order of the error term is optimal (by
an assertion of \cite{KK1}). Although we only compute  the top
order asymptotics of the singularity our method would allow a further
expansion (as in  \cite{KK1}), in fact  in principle a complete expansion.

In the bulk of the paper we  impose  Condition
 \ref{cond:one-body-starkPot} with the extra condition $q_2=0$.
 It is a minor technical issue to treat  the general case by using  
   the arguments of  the paper and the second resolvent equation (not
   to be elaborated on). The unique continuation
principle  \cite[Condition 2.4]{AIIS1} is fulfilled for  $q_2=0$  as
well as for  the example $q=\kappa r^{-1}$,   $d\geq 3$.

 It is convenient in this paper to reserve the notation  $\breve f$ to
 mean  any  smooth  convex function  $\breve f$ on $\R$ such
that $\breve f(t)=1$ for $t\leq 1/2$ and  $\breve f (t)=t$ for $t\geq
2$. This quantity it fixed throughout the paper and will be used
without reference.

We use the standard notation $\inp{z}=(1+\abs{z}^2)^{1/2}$ for $z$ in
a normed space, while we for any $m\in\N$  let
$\inp{z}_m=(m^2+\abs{z}^2)^{1/2}$ and $\hat z_m=z/\inp{z}_m$.
  We use the (standard) notation $F(x\in M)=1_M$ for the
characteristic function of a set $M$. 
 For any $\kappa\in\R\setminus\set{0}$ the notation
 $\chi(\cdot<\kappa)$ stands for any  smooth real function $\chi$ on
 $\R$ with
 $\chi'\in C_\c^\infty(\R)$ and   $\supp \chi \subseteq (-\infty,\kappa)$,
 and  for 
 $\kappa>0$ we require  $\chi(t)=1$ for
$t\leq 3\kappa/4$, while for  $\kappa<0$ we  require  $\chi(t)=1$ for
$t\leq 4\kappa/3$. Let 
$\chi(\cdot>\kappa)=\chi(-\cdot<-\kappa)$,
$\chi^\perp(\cdot<\kappa)=1-\chi(\cdot<\kappa)$ and $\chi^\perp(\cdot>\kappa)=1-\chi(\cdot>\kappa)$.

Let  
  $L^2_s=L^2_s(\R^d)=\inp{(x,y)}^{-s}\vH$ for  $s\in\R$ and
  $L^2_\infty=\cap_s L^2_s$. Let $H^2_\infty=H^2_\infty(\R^d)=
  (p^2+1)^{-1}L^2_\infty$. For $\varphi\in \vH$ and an operator $T$ on $\vH$ the notation
$\inp{T}_\varphi$  means $\inp{\varphi,T\varphi}$. The notation
$\vL(\vH)$ refers to the
set of bounded operators on $\vH$.

\section {Fourier-Airy transformation and the stationary phase method}\label{sec:Airy function asymptotics} 

One can   explicitly specify a diagonalizing transform $\int_{\R}
\oplus \vF_0(\lambda)\,\d \lambda $ such  that
$\delta(H_0-\lambda)=\vF_0(\lambda)^*\vF_0(\lambda)$, cf.  \cite{He, Ya1}. Writing the
Airy function by its  Fourier transform the expression for the kernel
of 
$\vF_0(\lambda)^*$ (it is not unique) is an oscillatory integral.  One
can then look at stationary points when applying the operator  to any $\xi\in
C^\infty_\c(\R^{d-1}_\zeta)\subseteq \Sigma$ and conclude  by stationary phase method
considerations, cf.  \cite {Ho} (or for example \cite  {II,St}), what the asymptotics
should be.  Below we state the results without giving details of
proof, deferring a more precise treatment  to Subsection \ref{subsec:The stationary phase method}. For
 simplicity we take below  $\lambda=0$ and consider 
\begin{align}\label{eq:free_eigFct}
  \begin{split}
\parb{\vF_0(0)^* \xi}(x,y)&=c \int \d \zeta \,\xi(\zeta)\int \e^{\i \theta}\,\d \eta;\\
c&=(2\pi)^{-\tfrac{d+1}2},\quad\theta=y \cdot \zeta-\eta^3/6+(x-\zeta^2/2)\eta.  
  \end{split}
\end{align} For the general case  we would replace $x$ by $x+\lambda$
in this expression.
 We look for critical points
 \begin{align*}
   &0=\partial_\eta \theta=-\eta^2/2-\zeta^2/2 +x \quad (\text{energy relation}),\\
&0=\partial_\zeta \theta=y-\eta\zeta \quad (\text{velocity relation}).
 \end{align*} Note that considering the momentum $\eta$ as an
 effective time indeed  the last equation
  written as  $\zeta=y/\eta$ is a `velocity relation'.  
  If we substitute  $\zeta=y/\eta$ in the
 argument of the function $\xi$, interchange order of integration and do the
 $\zeta$-integration we end up with 
 \begin{align*}
   &c\int \e^{\i \theta_1}\xi(y/\eta) \parb{\i \eta/(2\pi)}^{\tfrac{1-d}2} \,\d \eta;\\
&\theta_1(\eta)=y^2/(2\eta)-\eta^3/6 +x\eta.
 \end{align*} The critical  points of $\theta_1$
fulfill
\begin{align*}
  0=\partial_\eta \theta_1=-\tfrac 12 y^2/\eta^2-\tfrac 12\eta^2 +x,
\end{align*} which in turn fulfill
\begin{align*}
  {\eta^2 =x\pm \parb{x^2 - y^2}^{1/2}.}
\end{align*} We choose `$+$' (the case of `$-$' does  not contribute to the
asymptotics) leading to the two critical  points
\begin{align*}
  \eta=\pm \sqrt{x+ \parb{x^2 - y^2}^{1/2}}\approx \pm (2x)^{1/2}.
\end{align*}
The second order derivative is $ -\eta(1+\vO(y^2/x^2))$.
Whence the asymptotics of the integral is given as the sum of the 
following two terms (to be justified in Subsection \ref{subsec:The
  stationary phase method} and {Appendix \ref{Appendix}):}
\begin{align*}
 c\e^{ \pm\i \theta_{\rm ex}}&\xi(y/\eta) \parb{\i \eta/(2\pi)}^{-d/2}\\
&\approx \tfrac{\e^{\mp\i\pi d/4}}{\sqrt{2\pi}}(2x)^{-d/4}\e^{
  \pm\i\theta_{\rm ex}}\xi(\pm \omega);\\
 &\theta_{\rm ex}=\sqrt{x+ \parb{x^2 - y^2}^{1/2}}\parbb{\tfrac 12 y^2/\eta^2-\eta^2/6
  +x},\\
&\eta^2 =x+ \parb{x^2 - y^2}^{1/2},\quad \omega=(2x)^{-1/2 }y.
\end{align*} We may write the function $\theta_{\rm ex}$ as
\begin{align*}
  \theta_{\rm ex} =f^3_{\rm ex}/3:=  \tfrac 43\sqrt{x+ \parb{x^2 - y^2}^{1/2}}
 \parbb{x- \tfrac 12\parb{x^2 - y^2}^{1/2}}.
\end{align*} It is a solution to  the eikonal equation on the domain
of its definition, as
demonstrated in  \cite{AIIS2}, that is $\tfrac 12|\nabla\theta_{\rm ex}|^2 - x = 0$.

\subsection {The stationary phase method}\label{subsec:The stationary phase method} As in \eqref{eq:free_eigFct}   we
consider in this subsection double integrals of the form
\begin{subequations}
  \begin{align}\label{eq:fFunct}
  \begin{split}
    \phi_{\lambda,\tilde{a}}[\xi](x,y)=c\int \d
   \zeta \,\xi(\zeta)\int \e^{\i \theta_\lambda}\,\tilde a\,\d \eta;\\
c=(2\pi)^{-\tfrac{d+1}2},\quad \theta_\lambda&= y \cdot\zeta-\eta^3/6+(x+\lambda-\zeta^2/2)\eta.
  \end{split}
\end{align}  We call $\tilde{a}=\tilde{a}(x,y;\eta,\zeta)$ a
\emph{symbol}. By definition $ \tilde{a}$ is smooth in $(\eta,\zeta)$
and for some $n\in\N_0$, possibly depending on $\tilde{a}$,
\begin{align}
  \label{eq:1symb}
  \begin{split}
 \abs{\partial_{\eta,\zeta}^\alpha\,
   \tilde{a}}\leq
   C_{\alpha}(\zeta)\parbb{1+\tfrac {{\eta}^2}{\breve f(x)}}^{n}
 ;\quad 
   C_{\alpha}(\cdot ) \text{ locally bounded}.   
  \end{split}
\end{align}
  (We shall use  more refined  symbol classes in Sections \ref{subsubsec:Best results on the scattering
  matrix} and \ref{subsubsec:better Best results on the scattering
  matrix}.)  The function $\xi\in C_\c^\infty(\R^{d-1})$, and
$\lambda\in\R$. 
\end{subequations} Here  we  derive  the asymptotics {of} such integrals
as $\abs{(x,y)}\to \infty$.

The critical points of $\theta_\lambda$ are given by 
 \begin{align*}
   &0=\partial_\eta \theta_\lambda=-\eta^2/2-\zeta^2/2 +x+\lambda,\\
&0=\partial_\zeta \theta_\lambda=y-\eta\zeta.
 \end{align*}
{These equations suggest  two ways of integrating by
  parts based on the formulas
  \begin{subequations}
\begin{align}
  \label{eq:inteparts1i}\begin{split}
\e^{\i\theta_\lambda}&=\parbb{1+\parb{x+\lambda-\eta^2/2-\zeta^2/2}^2}^{-1}\\
&\phantom{{}={}}{}\cdot \parb{1-\i \parb{x+\lambda-\eta^2/2-\zeta^2/2}\partial_\eta}\e^{\i\theta_\lambda},
      \end{split}\\
\e^{\i\theta_\lambda}&=\parbb{1+\parb{y-\eta\zeta}^2}^{-1}\parb{1-\i \parb{y-\eta\zeta}\cdot\nabla_\zeta}\e^{\i\theta_\lambda}.
\label{eq:inteparts2i}
\end{align}    
  \end{subequations} 
First} we note that for $x \leq 2R$ for any $R>2$ (large), the function
$\phi_{\lambda,\tilde{a}}[\xi]$  has arbitrary polynomial decay in
$(x,y)$. This follows readily by repeated integration by parts using
\eqref{eq:inteparts1i} and \eqref{eq:inteparts2i}.

Pick for any sufficiently small  $\epsilon>0$  a function $\chi_\epsilon\in
C_\c^\infty((-2\epsilon,2\epsilon))$ such that  $\chi_\epsilon=1$ on
$(-\epsilon,\epsilon)$. Let $ \chi^\perp_\epsilon=1-\chi_\epsilon$,
and let $\chi_R=\chi(\cdot>R)$ for $R>2$.
  Then the functions
\begin{align*}
  &c\int \d
   \zeta \,\xi(\zeta)\int \e^{\i \theta_\lambda}\,\tilde
    a\,\chi^\perp_\epsilon\parb{\abs{\eta}-\sqrt{2x}}\chi_R(x)\,\d
    \eta,\\
&c\int \d
   \zeta \,\xi(\zeta)\int \e^{\i \theta_\lambda}\,\tilde
  a\,\chi_\epsilon\parb{\pm{\eta}-\sqrt{2x}}\chi^\perp_\epsilon\parb{\big|\zeta\mp
  y/\sqrt{2x}\big|}\chi_R(x)\,\d \eta,
\end{align*}  also  have  arbitrary polynomial decay in
$(x,y)$. This is seen by the same argument as for  $x \leq 2R$.

We conclude that the leading asymptotics  of
$\phi_{\lambda,\tilde{a}}[\xi]$ at infinity is the same as that of
$\phi^+_{\lambda,\tilde{a}}[\xi]+\phi^-_{\lambda,\tilde{a}}[\xi]$,
where
\begin{align}\label{eq:mainP}
  \phi^\pm_{\lambda,\tilde{a}}[\xi]=c\int \d
   \zeta \,\xi(\zeta)\int \e^{\i \theta_\lambda}\,\tilde
  a\,\chi_\epsilon\parb{\pm{\eta}-\sqrt{2x}}\chi_\epsilon\parb{\big|\zeta\mp
  y/\sqrt{2x}\big|}\chi_R(x)\,\d \eta.
\end{align} If $C>1$  and $\supp\xi\subset
\set{\abs{\zeta}<C/2}$, then $\phi^\pm_{\lambda,\tilde{a}}[\xi]$
vanish unless    ${\abs{y}< C\sqrt{2x}}$.   Moreover in the
support of the product of the $\chi$-factors in  the integrands of  \eqref{eq:mainP} the stationary
point is  uniquely given  for $R=R(\epsilon)$ big enough as 
\begin{align*}
  \eta^\pm=\pm \sqrt{x+\lambda+ \parb{(x+\lambda)^2 - y^2}^{1/2}}\mand
  \zeta^\pm=y/\eta^\pm,\text{ respectively}.
\end{align*}
Note  that uniformly in ${\abs{y}< C\sqrt{2x}}$,
\begin{align}\label{eq:asypFixed}
  \eta^\pm\mp \sqrt{2x}=\vO(1/\sqrt{2x})\mand \zeta^+=\pm \tfrac y{\sqrt{2x}}\parb{1+ \vO(1/x)}.
\end{align} 

Clearly we are left with the asymptotics in the set $\set{\abs{y}<
  C\sqrt{2x}}$ only.
 We  consider $\sqrt{2x}=h^{-1}$ as a large parameter,
write $\theta_\lambda=h^{-1}\,\tilde{\theta}_\lambda$ and then use
the stationary phase method. We can   in this way obtain  the
following leading order asymptotics as $h\to 0$.
\begin{align}\label{eq:AsySTA}
  \begin{split}
  \phi^\pm_{\lambda,\tilde{a}}[\xi](x,y)&=\tfrac{\e^{\mp\i\pi d/4}}{\sqrt{2\pi}}h^{d/2}\e^{
  \pm\i\theta_{\rm ex}(x+\lambda,y)}\xi(\pm h y)\tilde{a}(x,y;\pm h^{-1},\pm
  hy)+\vO\parb{h^{(d+2)/2}},\;\\&
\text{as } h=(2x)^{-1/2}\to 0 \text{ uniformly in }y\in \set{\abs{y}<
  Ch^{-1}}.   
  \end{split}
\end{align} 
The present  stationary phase problem has 
an additional 
 parameter dependence (that is  a {dependence on}  the variable
 $y$). Such problem  is rather
 standard, and we show 
 the result \eqref{eq:AsySTA} in Appendix \ref{Appendix}  by mimicking the
 proof of \cite[Theorem 4.3]{II} in the presence of the  additional 
 parameter. For very similar  results we
 refer to \cite[Theorems 7.7.5-6]{Ho}.

In particular, in  terms of    Besov spaces $\vB=\vB(f)$, $\vB^*=\vB^*(f)$
and $\vB_0^*=\vB_0^*(f)$,
defined for  the function $f$ given in \eqref{eq:par1} below,  we conclude from \eqref{eq:AsySTA} that
\begin{align}
  \label{eq:basyp}
  \begin{split}
\phi^+_{\lambda,\tilde{a}}[\xi],\,\phi^-_{\lambda,\tilde{a}}[\xi],&\,\phi_{\lambda,\tilde{a}}[\xi]\in
                                                     \vB^*,\\
                                                     \phi_{\lambda,\tilde{a}}[\xi](x,y)\quad\quad
                                                     &\\-\Sigma_\pm\tfrac{\e^{\mp\i\pi d/4}}{\sqrt{2\pi}}&h^{d/2}\e^{
  \pm\i\theta_{\rm ex}(x+\lambda,y)}\xi(\pm h y)\tilde{a}(x,y;\pm h^{-1},\pm
  hy)F(x>1)\in  \vB^*_0.  
  \end{split}
\end{align}

\section{Stationary scattering theory} \label{subsec:Results from
  cite{AIIS2}} 
We recall various  results of \cite{AIIS2}, proven
under  weaker conditions than Condition
\ref{cond:one-body-starkPot}.
\subsection{Parabolic coordinates and the  phase function $\theta^\lambda
  =f^3/3+\lambda f$} \label{subsec:Results from
  cite{AIIS2}0} It is  known since a long time  ago that
parabolic coordinates are useful for studying problems for Stark
Hamiltonians, see \cite{Ti}. Here we  recall the version of these
coordinates used in 
\cite{AIIS2}, which in turn is   a slight modification
 of the one of \cite{AIIS1}.
 We introduce $f\in C^\infty(\R^d)$ by the recipe 
\begin{subequations}
  \begin{align}
f(x,y)=\sqrt{\breve f (r+x)};\quad
r=(x^2+y^2)^{1/2}.
\label{eq:par1}
\end{align} Note that this function  obeys
\begin{align}\label{eq:conv}
  \nabla^2f^2\geq 0.
\end{align}

Of course a classical scattering orbit  will have $x>1$ eventually so that for
large time $f=(r+x)^{1/2}$. Since we will use the parabolic variable 
$(r+x)^{1/2}$ in quantum mechanics it is convenient to
introduce the above  regularization $f$, say also  named a \emph{parabolic
  variable}. We introduce 
other \emph{parabolic variables},
\begin{align}\label{eq:par2}
  g=y/f,\quad g_i=y_i/f;\quad i=2,\dots,d.
\end{align}
\end{subequations}

We recall a  few  calculus
formulas
in parabolic coordinates, cf.  \cite{AIIS2}. It is below  tacitly assumed
 that $r+x>2$. 
\begin{align}\label{eq:ort}
  \begin{split}
  f^2+g^2&=2r,\quad  f^2-g^2=2x,\quad f\abs{g}=\abs{y},\\ 2r \abs{\nabla f}^2&=1
,\quad 
\nabla f\cdot \nabla g_i=0;\quad i=2,\dots,d.  
  \end{split}
\end{align} 

Introducing
$\theta^\lambda =f^3/3+\lambda f$ for any  $\lambda\in\R$  (note the
superscript convention to distinguish this function and  the phase $\theta_\lambda$ of
\eqref{eq:fFunct}) we compute
(taking  here  $\lambda=0$)  
\begin{align}\label{eq:comtheta} 
  \begin{split}
\nabla \theta^0 &=\tfrac 1{2r}(f^3,fy),\\
(\nabla^2\theta^0)^{11} &= -\tfrac 12 \tfrac{xf^3}{r^3}+\tfrac 34 \tfrac{f^3}{r^2}, \\
(\nabla^2\theta^0)^{1\alpha} &= -\tfrac 12 \tfrac{y^\alpha
                             f^3}{r^3}+\tfrac 34 \tfrac{y^\alpha f}{r^2}, \\
  (\nabla^2\theta^0)^{\alpha\beta} &= -\tfrac 12
                                   \tfrac{y^{\alpha}y^{\beta}f}{r^3}+\tfrac
                                   14
                                   \tfrac{y^{\alpha}y^{\beta}}{r^2f}+\tfrac
                                   12
                                   \tfrac{f}{r}\delta^{\alpha\beta},\\
\Delta\theta^0 &= \tfrac d2 \tfrac fr.
  \end{split}
\end{align} Moreover
\begin{align}\label{eq:gradComP}
  \nabla \parb{ \tfrac{2r}{f^2}} =\tfrac 2{rf^4}(-y^2,xy).
\end{align}
Letting  $T$ denote the change to parabolic coordinates, $ T(x,y)=(f,g)$,
  then a computation using \eqref{eq:ort} shows that
\begin{align}\label{eq:Jac1}
  J:=\abs {\det T'}= \tfrac {f^{2-d}}{f^2+g^2}.
\end{align} 

  Using again \eqref{eq:ort}  we easily compute the partial derivative
  with respect to $f$
\begin{align}\label{eq:parF}
  \partial_f=\abs{\nabla f}^{-2}\nabla f\cdot \nabla=F\cdot \nabla;\quad F:=\tfrac
  {2r}{f^2}\nabla \theta^0=2r\nabla f,
\end{align} and by using \eqref{eq:comtheta}--\eqref{eq:Jac1}  we compute
\begin{align}\label{eq:Jac2}
  \partial_f\ln\parb{J^{-1/2}}=\tfrac 12 \Div  F.
\end{align}

In the more restricted  region $\set{x>1,x>2\abs{y}}$ (for example)  the following
uniform bounds hold,   cf.  \cite[Lemma 3.4]{AIIS2},
\begin{align}\label{eq:comfs}
  \begin{split}
\theta^0&=\tfrac{(2x)^{3/2}}{3}\parb{1+\tfrac
  38\abs{y/x}^2+\vO\parb{\abs{y/x}^4}},\\
\theta_{\rm ex}-\theta^0&=f^3\vO\parb{\abs{y/x}^4},\\
\nabla \theta_{\rm ex}-\nabla \theta^0&=f \vO\parb{\abs{y/x}^3}, \\
\nabla^2 \theta_{\rm ex}-\nabla^2 \theta^0&=f^{-1} \vO\parb{\abs{y/x}^2}, \\
 f-f_{\rm ex}&=f\vO\parb{\abs{y/x}^4}.  
  \end{split}
\end{align} 

We  can  obtain similar formulas for $\lambda\neq 0$
by using  a Taylor  expansion in combination with  \eqref{eq:comtheta}
and \eqref{eq:comfs}. Thus in particular in the same restricted  region
\begin{align}
  \label{eq:difEst} \theta_{\rm ex}(x+\lambda,y)-\theta^\lambda(x,y)=f^3\vO\parb{\abs{y/x}^4}+f\vO\parb{\abs{y/x}^2}+f^{-1}\vO\parb{\abs{y/x}^0}.
\end{align}

\subsection{Stationary wave
    operators and the scattering matrix} \label{subsec:Stationary wave
    operators and the scattering matrix} 
 We introduce    the \emph{radiation operators} 
\begin{align}\label{eq:QM11}
  \gamma^{\lambda\pm}=p\mp\nabla
  \theta^\lambda \quad \mand
  \quad \gamma_{\|}^{\lambda\pm}=\Re\parb{ 2r\nabla f\cdot \gamma^{\lambda\pm}}.
\end{align}
 Similarly we denote $\Gamma^{\lambda\pm}=(\gamma^{\lambda\pm},y/f^2)$ and use  the notation $\Gamma^{\lambda\pm}_k$ for its  components, although there is a
  $\lambda$-dependence only for $k\leq d$.

Recall the notation $\vB=\vB(f)$, $\vB^*=\vB^*(f)$ and $\vB_0^*=\vB_0^*(f)$  for  the Besov
spaces. The following   radiation condition bounds
are  simplified versions of similar bounds  stated and proved in \cite{AIIS2}.

\begin{proposition}\label{prop:phase-space-radiation} Let $\epsilon>0$
  and $\psi\in L^2_\infty$ be given. Let $\phi=R(\lambda\pm\i
  0)\psi$ 
  for $\lambda\in \R$ ($ \phi \in \vB^*$ by \cite{AIIS1}). Then for
  all $k,l=1,\dots, 2d-1$ and  locally
  uniformly in $\lambda$
  \begin{subequations}
 \begin{align}\label{eq:p1}
   &\norm{f^{(1+2\delta-\epsilon)/2}\,\Gamma^{\lambda\pm}_k\phi}_{\vB^*}<\infty,\\\label{eq:p2}
&\norm{f^{1+2\delta-\epsilon}\,\Gamma^{\lambda\pm}_k\,\Gamma^{\lambda\pm}_l\phi}_{\vB^*}<\infty,\\\label{eq:p3}
&\norm{f^{1+2\delta-\epsilon}\gamma_{\|}^{\lambda\pm}\phi}_{\vB^*}<\infty.
\end{align}   
  \end{subequations}
  \end{proposition}

The stationary  wave operators, the scattering matrix and the  generalized
eigenfunctions from \cite{AIIS2} are constructed in terms of the phase
$\theta^\lambda$ as follows.

We introduce, using  
 parabolic coordinates,
\begin{align*}
   \vR^{\lambda\pm}_{\psi,f} (\zeta):=\parb{J^{-1/2}\e^{\mp\i\theta^\lambda}
R(\lambda \pm \i 0)\psi}(f,\pm \zeta);\quad  \psi\in \vB.
\end{align*} 
 It is an almost trivial consequence of  \eqref{eq:parF}, \eqref{eq:Jac2} and Proposition 
\ref{prop:phase-space-radiation}, that  for any $\psi\in
L^2_\infty\subseteq \vB$ there exist $\vSigmalim_{f\to
  \infty} \,\vR^{\lambda\pm}_{\psi,f} $. Indeed the computation
\begin{align*}
\partial_f \vR^{\lambda\pm}_{\psi,f} 
= 
\i \parb{J^{-1/2}\mathrm e^{\mp {\mathrm i}\theta^\lambda}\gamma_{\|}^{\lambda\pm}
R(\lambda\pm {\mathrm i}0)\psi}(f, \pm \cdot), 
\end{align*}\cas  and \eqref{eq:p3} show that  the functions $[0,\infty)\ni
f \to \vR^{\lambda\pm}_{\psi,f} \in \Sigma$ have  integrable derivatives.
Consequently  there exist  \emph{wave operators at fixed energy},
\begin{align*}
  \vF^\pm(\lambda)\psi:=\tfrac{\e^{\pm\i\pi (d-2)/4}}{\sqrt{2\pi}}
\vSigmalim_{f\to \infty} \,\vR^{\lambda\pm}_{\psi,f};\quad \psi\in
L^2_\infty.
\end{align*}
 These obey   the formulas
\begin{align}
  \label{eq:fund}
  \|\vF^\pm(\lambda)\psi\|^2=\inp{\psi, \delta(H-\lambda)\psi};\quad \delta(H-\lambda)=\pi^{-1}\mathop{\mathrm{Im}}R(\lambda+\i 0).
\end{align}

It follows from \eqref{eq:fund} that 
$\norm{\vF^+(\lambda)\psi}=\norm{\vF^-(\lambda)\psi}$ and that
$\vF^\pm(\lambda)\psi\in \Sigma$ are  defined for  $\psi\in
\vB$. This extension of $\vF^\pm(\lambda)$ is given explicitly as
follows. For  vector-valued functions $\xi$ on $\R$ (or
on $\R_+$) we
use the notation $\intR \xi(r)\,\d r=\rho^{-1}\int_{\rho}^{2\rho} \xi(r)
\,\d r$, $\rho>0$. Then for  any $\psi\in \vB$ the vectors
$\vF^\pm(\lambda)\psi $ are given as the 
averaged limits  
\begin{align}\label{eq:extbF2}
  \begin{split} \vF^\pm(\lambda)\psi =\vSigmalim_{\rho\to
    \infty} {-\!\!\!\!\!\!\int_\rho}  \,\tfrac{\e^{\pm\i\pi
      (d-2)/4}}{\sqrt{2\pi}}\,\vR^{\lambda\pm}_{\psi,f}\,\d f. 
\end{split}
\end{align} The maps $\vB\ni \psi\to
\vF^\pm(\lambda)\psi\in \Sigma$ are surjective (cf. Theorem
\ref{thm:char-gener-eigenf-1} \label{sec:stat-wave-oper}\ref{item:14.5.14.4.17}
stated below), and it is also proven in \cite{AIIS2}
that  for any $\psi\in \vB$ the maps
  $\R\ni\lambda\to \vF^\pm(\cdot)\psi\in \Sigma$ are continuous.

Consequently we can
define  the \emph{scattering
  matrix}  as the unique unitary operator $S(\lambda)$ on
$\Sigma$ obeying
\begin{align}\label{eq:scattering_matrix}
  {\vF^+(\lambda)\psi}=S(\lambda){\vF^-(\lambda)\psi},
\end{align} and deduce that 
 the map $\R\ni\lambda\to S(\lambda)\in \vL(\Sigma)$ is
strongly continuous.

Along with $ \mathcal H=L^2(\R^d)$ we introduce the space 
 \begin{align*}
  \widetilde \vH =L^2(\R, \d \lambda;\Sigma),
\end{align*} 
  and let $M_\lambda$ be the operator of multiplication by $\lambda$ on
$\widetilde\vH$.
We introduce the operators 
\begin{align*}
  \vF^\pm=\int_{\R} \oplus \vF^\pm(\lambda)\,\d \lambda\colon
  \vB\to C(\R;\Sigma).
\end{align*}

\begin{thm}
  \label{prop:dist-four-transf} 
The operators   $\vF^\pm$ extend uniquely as to become  unitary
operators  $\vF^\pm:\vH\to \widetilde{\mathcal H}$. These  extensions 
satisfy  $\vF^\pm H=
M_\lambda \vF^\pm$.
\end{thm}

 The extensions  $\vF^\pm$ in  Theorem
  \ref{prop:dist-four-transf} are called \emph{stationary wave
    operators}, and the  first assertion of the theorem may be  referred to as
    \emph{stationary completeness}. On the other hand the adjoint of the operators   $\vF^\pm(\lambda)$, i.e.   $\vF^\pm(\lambda)^*\in \vL(\Sigma, \vB^*)$, are called \emph{stationary wave
    matrices}.

\subsection{Minimal generalized eigenfunctions} \label{Miminum generalized eigenfunctions}

 For any $\xi\in \Sigma$ we
introduce purely outgoing/incoming  approximate  generalized eigenfunctions  
{$\phi^{\lambda\pm}[\xi]\in \vB^*$} by, using the parabolic coordinates,
\begin{align}\label{eq:gen1B}
\begin{split}
\phi^{\lambda \pm}[\xi](f,g)= \tfrac{\e^{\mp\i\pi d/4}}{\sqrt{2\pi}}\chi^{\perp}(f<2)J^{1/2}(f,g)\,\e^{\pm\i\theta^\lambda(f,g)}\,\xi(\pm
g).
\end{split}
\end{align}  
 These functions may be seen as   purely
 outgoing/incoming (zeroth order) WKB-approximations of  generalized
 eigenfunctions. In fact for $\xi\in C^\infty_{\mathrm c}(\R^{d-1})$ we can compute 
 $\psi^{\lambda \pm}[\xi] :=(H-\lambda)\phi^{\lambda \pm}[\xi] \in \vB,$
  which allows us to  consider the exact solutions 
\begin{subequations}
  \begin{align}\label{eq:rep1}
    \begin{split}
    \phi_{\rm{ex}}^{\lambda \pm}[\xi]:= \phi^{\lambda \pm}[\xi]-R(\lambda\mp\i 0)\psi^{\lambda \pm}[\xi];\quad\xi\in
    C_\c^\infty(\R^{d-1}).  
    \end{split}
  \end{align}
   Furthermore for $\xi\in C^\infty_{\mathrm c}(\R^{d-1})$ we have the
   formulas 
  \begin{align}\label{eq:exForm}
    {\phi_{\rm{ex}}^{\lambda\pm}[\xi]}&=\vF^\pm(\lambda)^* \xi,\\
\label{eq:0Sommer}
    0&= \phi^{\lambda \pm}[\xi]-R(\lambda\pm\i 0)\psi^{\lambda \pm}[\xi],\\
\label{eq:rep3}
    \xi&= \pm\i 2\pi \vF^\pm(\lambda)\psi^{\lambda \pm}[\xi]. 
  \end{align}  Here \eqref{eq:0Sommer} is a consequence of the  
Sommerfeld uniqueness result of \cite{AIIS1}, and obviously
in turn  \eqref{eq:rep3} is a consequence of \eqref{eq:0Sommer}.
\end{subequations}

 The elements of the space
\begin{equation*}
    \vE_\lambda:=\{\phi\in \vB^*\,|\, (H-\lambda)\phi=0\}
  \end{equation*}  are  called
  \emph{minimal generalized eigenfunctions}. They are all of the form
  specified to the right in
\eqref{eq:exForm} with $\xi \in \Sigma$, as stated in the following theorem from \cite{AIIS2}.

\begin{thm}
  \label{thm:char-gener-eigenf-1}
\begin{subequations}
  \begin{enumerate}[1)]
  \item\label{item:14.5.13.5.40} For any one of $\xi_\pm \in \Sigma$ or $\phi\in \vE_\lambda$ the
    two  other  quantities in the triple $(\xi_-,\xi_+,  \phi)$ uniquely exist
    such that
    \begin{align}\label{eq:gen1}
      \phi -\phi^{\lambda +}[\xi_+]-\phi^{-}[\xi_-]\in
      \vB_0^*.
    \end{align}

  \item \label{item:14.5.13.5.41} The  correspondences  in  \eqref{eq:gen1} are  given  by  the
    formulas
    \begin{align}\label{eq:aEigenfw}
     \phi&= \vF^\pm(\lambda)^*\xi_\pm,\quad \xi_+=S(\lambda)\xi_-,\\
\xi_\mp&=\mp\tfrac1{2} \tfrac{\sqrt{2\pi}}{\e^{\pm\i\pi d/4}}\vSigmalim_{\rho\to \infty}
-\!\!\!\!\!\!\int_\rho \,\Bigl ( J^{-1/2}
\parbb{f\sqrt{2r}}^{-1}\,\e^{\pm\i\theta^\lambda}\gamma_{\|}^{\lambda\pm}\phi\Bigr)(f,\mp\cdot)\,\d f.\label{eq:aEigenfwB}
    \end{align}
    In  particular  the  wave  matrices  $\vF^\pm(\lambda)^*\colon\Sigma\to
    \vE_\lambda$ are linear isomorphisms.

  \item\label{item:14.5.13.5.42}  The wave matrices
    $\vF^\pm(\lambda)^*\colon\Sigma\to \vE_\lambda\,(\subseteq \vB^*)$
    are bi-continuous. In fact
    \begin{align}\label{eq:aEigenf2w}
      \|\xi_\pm\|_{\Sigma}^2=\pi\lim_{m \to \infty}2^{-m}\norm{F(2^m\leq f< 2^{m+1})\phi}^2.
    \end{align}

  \item\label{item:14.5.14.4.17}   The   operators    $\vF^\pm(\lambda)\colon   \vB\to   \Sigma$   and
    $\delta(H-\lambda)\colon \vB\to \vE_\lambda$ map  onto.
  \end{enumerate}
   \end{subequations}
\end{thm}

The assertions 
\eqref{eq:aEigenfw} and \eqref{eq:aEigenfwB} provide a general
formula for  $S(\lambda)\xi$ (take $\xi=\xi_-$ and $\phi=
\vF^-(\lambda)^*\xi$). However for smooth compactly supported $\xi$
there are the following alternative recipes  for calculating the scattering matrix.   

    \begin{corollary}\label{cor:scatt-matr-gener} For  any $\xi \in  C^\infty_{\mathrm c}(\R^{d-1})$
      \begin{align}\label{eq:formS1}
S(\lambda)\xi=-\i 2\pi \vF^+(\lambda)\psi^{\lambda -}[\xi]=-\tfrac{\sqrt{2\pi}}{\e^{-\i\pi d/4}} \vGlim_{\rho\to
    \infty} \intR \vR_{\psi^{\lambda-},f}^{\lambda+}\,\d f.
        \end{align}

 For any $\xi, {\xi'} \in
C^\infty_{\mathrm c}(\R^{d-1})$ 
\begin{align}\label{eq:formS2}
  \begin{split}
    \tfrac 1{2\pi \i}\inp{{\xi'},S(\lambda)\xi}=\inp{\psi^{\lambda+}[{\xi'}],R(\lambda+
      \i
      0)\psi^{\lambda-}[\xi]}-\inp{\phi^{\lambda+}[{\xi'}],\psi^{\lambda-}[\xi]}.
\end{split}
\end{align} 
    \end{corollary}

There are other representations like
\eqref{eq:formS2} of the scattering matrix, which we will state and examine in
this paper (see Sections \ref{subsubsec:Best results on the scattering
  matrix} and \ref{subsubsec:better Best results on the scattering
  matrix}).

We also note that our notation appears  consistent in the  case
$q=0$ in the following sense. Recall that $\vF_0(\lambda)$ is given by
\eqref{eq:free_eigFct} (with $x\to x+\lambda$).

\begin{corollary}\label{cor:scatt-consis}  For $q=0$
  \begin{align*}
    \vF_0(\lambda)=\vF^+(\lambda)=\vF^-(\lambda).
  \end{align*} In particular $S(\lambda)=I$ in this case.
\end{corollary}
  \begin{proof}  We prove the first identity only. The proof of the
    identity $\vF_0(\lambda)=\vF^-(\lambda)$ is similar, so this
    suffices. In turn it suffices to show that $ \vF_0(\lambda)^*\xi=
    \vF^+(\lambda)^*\xi$ for any $\xi\in C_\c^\infty(\R^{d-1})$.
    We write 
    $\vF_0(\lambda)^*\xi=\phi_{\lambda,1}[\xi]$ in agreement with   \eqref{eq:fFunct}. By \eqref{eq:basyp}
    and Theorem
\ref{thm:char-gener-eigenf-1} \ref{item:14.5.13.5.40}--\ref{item:14.5.13.5.41}  we may write
$\phi_{\lambda,1}[\xi]=\vF^+(\lambda)^*\check{\xi}$ for a unique 
$\check{\xi}\in \Sigma$, which may be computed by
\eqref{eq:aEigenfwB}. For this end we first compute  
\begin{align*}
  &\tfrac1{2} \tfrac{\sqrt{2\pi}}{\e^{-\i\pi d/4}}J^{-1/2}
\parbb{f\sqrt{2r}}^{-1}\,\e^{-\i\theta^\lambda}\gamma_{\|}^{\lambda-}\phi_{\lambda,1}[\xi]=\tfrac{\sqrt{2\pi}}{\e^{-\i\pi d/4}}J^{-1/2}\,\e^{-\i\theta^\lambda}\phi_{\lambda,\tilde{a}}[\xi];\\
\tilde{a}&=\tfrac1{2} 
\parbb{f\sqrt{2r}}^{-1} \parbb{2r\nabla f\cdot\parb{(\eta,\zeta)+\nabla
           \theta^\lambda}-\tfrac {\i}2 \Delta \parb{2r\nabla f}}.
\end{align*} By using \eqref{eq:basyp} to  this $\tilde{ a}$, 
\eqref{eq:difEst}  and \eqref{eq:aEigenfwB} it
follows that  indeed $ \check{\xi}=\xi$.
\end{proof}

\section{Identificaton of wave
  operators and  scattering matrices} \label{subsec:Identificaton of wave operators}

The \emph{time-dependent wave
  operators} are given by 
\begin{align*}
  W^{\pm}:=\slim_{t\to \pm\infty}\e^{\i tH}\e^{-\i tH_0},
\end{align*} and the corresponding \emph{scattering
  operator} $S=(W^+)^* W^-$ commute with $H_0$, yielding the   representation 
\begin{align}\label{eq:scartTi}
  \vF_0S\vF_0^{-1}=\int _\R \oplus \,S(\lambda)\,\d \lambda.
\end{align} In this section we show that 
\begin{align}\label{eq:w10}
  \vF_0(W^\pm)^*=\vF^\pm,
\end{align}  which  implies  that the (almost everywhere defined)  operator
  $S(\lambda)$ in \eqref{eq:scartTi} is equal to  the (everywhere defined
  strongly continuous) scattering matrix
of Subsection 
\ref{subsec:Stationary wave
    operators and the scattering matrix}.

It follows from the Avron-Herbst formula 
\cite{AH} that for any $\varphi\in \vH$
\begin{align*}
  \parb{\e^{-\i t H_0}\varphi}(x,y)&= \e^{-\i\pi d/4} t^{-d/2} \e^{\i\{ -t^3/6 +tx +
  [(x-t^2/2)^2+y^2]/(2t)\}}
                                     \hat\varphi((x-t^2/2)/t,y/t)\\&\quad\quad\quad
                                                                     +o_{\vH}(\abs{t}^0)
  \text{ as }\abs{t}\to \infty;
\end{align*} here by definition $\norm{o_{\vH}(\abs{t}^0)}_{\vH}\to 0$
for $\abs{t}\to \infty$.

To identify the wave
  operators it would  be tempting, based on Section \ref{subsec:Results from
  cite{AIIS2}},  to try to compute directly  the 
$L^2$-asymptotics of  integrals of the form
\begin{align*}
  \int \e^{-\i t \lambda}J^{1/2}\e^{  \pm \i
  \theta^\lambda}\xi(\pm y/f)h(\lambda)\,\d \lambda,
\end{align*}  where $\xi\in C_\c^\infty (\R^{d-1})$ and  $h\in
C_\c^\infty (\R)$ and compare with the right-hand side of the above
formula,  cf. \cite{ IS3}. However such  computation does not seem doable.
 We proceed differently introducing, cf. \eqref{eq:fFunct} and \eqref{eq:1symb},
\begin{align}\label{eq:free_appr_eigFct}
  \begin{split}
\parb{ F^{\pm}_{0}(\lambda)^* \xi}(x,y)&=c\int \d \zeta \,\xi(\zeta)\int \e^{\i \theta_\lambda}\chi_\pm(\eta)\,\d \eta;\\
c=(2\pi)^{-\tfrac{d+1}2}&,\quad \theta_\lambda= y \cdot\zeta-\eta^3/6+(x+\lambda-\zeta^2/2)\eta,  
  \end{split}
\end{align} in terms of  the partition $\chi_++\chi_-=1$,
$\chi_+=\chi(\cdot>1)$ and $\chi_-=\chi^\perp(\cdot>1)$. 
 Using the formula
 \begin{align}\label{eq:dereta}
   (H_0-\lambda)\e^{\i \theta_\lambda}=-(\partial_\eta \theta_\lambda)\e^{\i
   \theta_\lambda}=\i \partial_\eta \e^{\i
   \theta_\lambda},
 \end{align} we can integrate by parts and deduce  that 
 \begin{align*}
    g(\lambda):=(H_0-\lambda)  F^\pm_0(\lambda)^* \xi=-\i c\int \d \zeta \,\xi(\zeta)\int \e^{\i \theta_\lambda}\chi'_\pm(\eta)\,\d \eta.
 \end{align*} By  the method of non-stationary phase,
 cf.   the first
  part of Subsection \ref{subsec:The stationary phase method},  this integral
has polynomial  decay  as $r=\abs{(x,y)}\to \infty$. In particular $g$
is an $L^2_1$-valued function, in fact (by the same argument)
\begin{align}\label{eq:Ceen}
  g\text{ is  }C^1 \text{ as an  } L^2_1\text{-valued function}.
\end{align}

Next by using  \eqref{eq:basyp}, \eqref{eq:difEst} and  Theorem
\ref{thm:char-gener-eigenf-1} \ref{item:14.5.13.5.40}  we deduce (cf. the proof
of  
\cite[Lemma 4.4]{AIIS2})
that the  generalized eigenfunctions in
the
formula \eqref{eq:rep1}
are given  by 
\begin{align}\label{eq:exacGen2}
\phi_{\rm{ex}}^{\lambda\pm}[\xi]  = F^\pm_0(\lambda)^* \xi- R(\lambda\mp\i 0)(H-\lambda) F^\pm_0(\lambda)^* \xi. 
\end{align}

 It is  easy to analyze the integral (focusing below on the case $t\to +\infty$)
\begin{align*}
  I^+(t):=\int \e^{-\i t \lambda}\parb{ F^+_0(\lambda)^*\xi}
  h(\lambda)\,\d \lambda;\quad t>0,\,\xi\in C_\c^\infty (\R^{d-1}),\,h\in
C_\c^\infty (\R).
\end{align*} In fact by the  method of non-stationary phase
\begin{align}\label{eq:Iint}
  I^+(t)+o_{\vH}({t}^0)= \int \e^{-\i t \lambda}\parb{\vF_0(\lambda)^*\xi}
  h(\lambda)\,\d \lambda={\e^{-\i t H_0}\varphi},  
\end{align} where  $\varphi\in \vH$ is fixed by $\vF_0\varphi=h\otimes \xi\in \tilde{\vH}$.

 In the paper \cite {II} completeness for Schr\"odinger
operators is considered/proven  from
the stationary point of view. In our setting one would look at the 
exact solution to the Schr\"odinger equation
\begin{align}\label{eq:exactSol}
   \int \e^{-\i t \lambda}\parb{\vF^+(\lambda)^* \xi}\,h(\lambda)\,\d \lambda=\int \e^{-\i t \lambda}\phi_{\rm{ex}}^{\lambda+}[\xi]h(\lambda)\,\d \lambda.
\end{align} Using \eqref{eq:Ceen}, \eqref{eq:exacGen2}, 
\eqref{eq:Iint}  and \cite [Lemma 5.1]{II} we
obtain that the wave packet \eqref{eq:exactSol} is of the form
\begin{align*}
  \e^{-\i t H_0}\varphi+o_{\vH}({t}^0)=\e^{-\i t H}W^+\varphi+o_{\vH}({t}^0)
\end{align*}
 with the above $\varphi\in \vH$. Note at this point that the condition $\int
 _0^\infty \norm{\widehat{hg}(s)}_{\vH}\,\d s<\infty$ of \cite [Lemma
 5.1]{II} is fulfilled due to \caS, the Plancherel theorem  and \eqref{eq:Ceen}. Since $\e^{-\i t H}W^+\varphi$ is
 also an exact solution to the Schr\"odinger equation it follows from
 the unitarity property of the Schr\"odinger propagator that
 $\vF^{+*}\parb{h\otimes \xi}=W^+\varphi$. Consequently (by density)
$\vF^{+*}\vF_0=W^+$, which is the `plus case' of \eqref{eq:w10}. The
`minus  case' of \eqref{eq:w10} can be derived similarly.

We  learn from \eqref{eq:w10} that there exist
the
\emph{asymptotic orthogonal momenta}
\begin{align*}
  p_y^{\pm} =\lim_{t\to \pm \infty} \e^{\i tH}p_y \e^{-\i tH}=\parb{\vF^\pm}^*
  \parbb{\int _\R \oplus M_{\zeta}\,\d \lambda}\vF^\pm;
\end{align*} here the limit is taken in the strong resolvent sense and
$ M_{\zeta}$ denotes multiplication by (the components of) $\zeta$ on $\Sigma=L^2\parb{\R^{d-1}_\zeta}$.
 Whence formally the (Schwartz) kernel 
 $ S(\lambda)(\zeta,\zeta')$ of the  scattering matrix is defined in
 terms of incoming and outgoing asymptotic orthogonal momenta, cf. the
 discussion in Section \ref{sec:Introduction}.

\section{Resolvent bounds} \label{subsubsec:Resolvent bounds}

We recall the following elementary result  from \cite{ AIIS1}.
\begin{align}\label{eq:neg2}
   \forall\alpha\in\N_0^d\text{ with }\abs{\alpha}\leq 1\,\forall k\in \N\,\forall h\in C^\infty_\c(\R):\, p^\alpha\abs{x}^{k/2}\chi(x<0)h(H)\in \vL(\vH).
  \end{align}

Let $A_m=\Re\parb{ \nabla
   f_m\cdot p}$ where $f_m(x,y)=
\sqrt{\breve f(2x+2\inp{y}_m)}$ with $\breve f$ given as in
\eqref{eq:par1} and with   $m\in \N$; recall the notation   $\inp{y}_m=(m^2+\abs{y}^2)^{1/2}$. Let    $\tilde A_{m}=\tfrac 12\Re\parb{ \nabla
  f_m^2\cdot p}=f_m^{1/2}A_mf_m^{1/2}$.

We compute
\begin{align}\label{eq:mourr}
  \i [H,2\tilde A_{m}]= p\cdot\parb{\nabla^2f_m^2}p+ \partial_xf_m^2-\parb{\nabla  f_m^2}\cdot
  \nabla q
  -\tfrac14\Delta^2 f_m^2,
\end{align}
which 
leads to  
\begin{align*}
  \i [H,2\tilde A_{m}]& \geq 2-\parb{\nabla  f_m^2}\cdot
  \nabla q
  -\tfrac14\parb{\Delta^2 f_m^2}-C_1F(2x+2\inp{y}_m\leq 2)\\
&\geq  2-C_2 \tfrac 1 {\inp{r}}
  -C_3m^{-3}+C_4\tfrac {x-1}m F(x-1+m\leq0).
\end{align*}  
  In combination with  \eqref{eq:neg2}  we conclude that for
  any given energy the  Mourre estimate (see \cite{Mo}) holds for $\tilde A_{m}$
  with a constant as close
to $1$ as  wished provided we take $m$ large enough. 

For convenience we abbreviate $A_{m}=A$ and $
f_m=f$, noting  that this $f$  is different from \eqref{eq:par1} used previously.
Now  the following estimates hold locally uniformly in
$\lambda\in\R$, cf. the multiple commutator methods of \cite{GIS,   AIIS3}. The parameter $m$ may
 depend on $\lambda$, however  locally it can be taken independently
 of $\lambda$,  and it may depend on the 
 parameters $t,t'$  appearing in the estimates   (however the
 dependence is only on $\kappa\in (0,1)$ provided  $t,t'\in
 [\kappa-1,1-\kappa]$; in our
 application we consider fixed parameters {only).}
\begin{subequations}
\begin{align}\label{eq:LAPi}
  f^{-s}R(\lambda\pm \i 0)f^{-s}\in
    \vL(\vH);\quad s>1/2.
\end{align}
\begin{align}\label{eq:twosidedai}
  \begin{split}
    \chi(\pm A<t)f^sR(\lambda\pm \i 0)&f^{-1-s}\in
    \vL(\vH);\\&\quad s>-1/2,\,t<1.
\end{split}
\end{align}
\begin{align}\label{eq:twosidedbi}
\begin{split}
\chi(\pm A<t)f^sR(\lambda\pm \i 0)f^{s}\chi(\pm A>t')&\in
    \vL(\vH);\\&\quad s\in\R,\,-1<t<t'<1.
  \end{split}
\end{align}
\end{subequations}
\begin{align}\label{eq:LAP2i}
  \forall k\in \N:\quad f^{-s}R(\lambda\pm \i 0)^kf^{-s}\in
    \vL(\vH);\quad s>k-1/2.
\end{align} The {last}  estimate is a consequence of
\eqref{eq:LAPi}--\eqref{eq:twosidedbi} and an   algebraic argument
(cf. \cite{ Is,Je}), in
fact there are `microlocal bounds' in the spirit of
\eqref{eq:LAPi}--\eqref{eq:twosidedbi} for powers also ({deducible from} 
the same argument). Such  estimates would be  useful for obtaining  regularity
of the $S$-matrix in  the spectral parameter, however  this topic will not
be studied in the  paper.

\section{Classical mechanics bounds and transport equations} \label{subsubsec:Classical
  mechanics bounds}
We may associate  to the operator $A=A_{m}$ of the previous section the  `symbol'
\begin{align*}
  a=a_m&=\tfrac{\eta+\hat y_m\cdot\zeta}{f_m} ;\\
  \inp{y}_m&=(m^2+\abs{y}^2)^{1/2},\,\,\hat y_m=y/\inp{y}_m,\quad f_m(x,y)=
\sqrt{\breve f(2x+2\inp{y}_m)}.
\end{align*}
 Here $m$ is a fixed large positive
  integer, and  by definition  $f_m=\sqrt{2x+2\inp{y}_m}$ for
  $x+\inp{y}_m>1$. Let $ \breve a=\tfrac{\eta+\hat
    y_m\cdot\zeta}{\sqrt{2x+2\inp{y}_m}}$ for $x+\inp{y}_m>0$.
 
We consider for any such $m$ and  for any $\varepsilon\in(0,1)$
 \begin{align*}
   \vX^\pm_\varepsilon=\vX^\pm_{m,\varepsilon}:=\set[\big]{x+\inp{y}_m>0,\quad \pm \breve a>-\varepsilon}.
 \end{align*}
 \begin{lemma}\label{lemma:free-class-mech} The sets
   $\vX^+_\varepsilon$ and $\vX^-_\varepsilon$ are  
 preserved by the free classical forward and
 backward flow $\Theta$ given  by \eqref{eq:freeClas} with  $t\geq 0$
 and $t\leq 0$, respectively.
   \end{lemma}
   \begin{proof}
   We estimate as follows on
 $\vX^\pm_\varepsilon$ for $\pm t\geq 0$,
  \begin{align}\label{eq:BEST}
    \begin{split}
     2x(t)+2\inp{y(t)}_m&=t^2 +2t(\eta+\hat y_m\cdot
     \zeta)+(2x+2\inp{y}_m)\\&\quad +2\parbb{\sqrt{2ty\cdot \zeta+\inp{y}_m^2+(t\zeta)^2}-t\hat y_m\cdot
     \zeta-\inp{y}_m}\\
&\geq t^2 -2\abs{t}\varepsilon \sqrt{2x+2\inp{y}_m}+(2x+2\inp{y}_m)\\
&\geq (1-\varepsilon)\parb{t^2+2x+2\inp{y}_m}. 
    \end{split}
  \end{align} In particular the left-hand side stays positive. The
 {`symbol'}  $ \breve a$ is well-defined  on
  $\vX^+_\varepsilon\cup \vX^+_\varepsilon$, and 
by a   free classical Mourre estimate, cf. the calculation \eqref{eq:mourr},
  \begin{align*}
    \tfrac{\d}{\d t}
     \breve a(t)\geq \parb{1- \breve a(t)^2}/\sqrt{2x(t)+2\inp{y(t)}_m}\quad \text{on
    }\vX^\pm_\varepsilon \text{  for 
    }\pm t\geq0.
  \end{align*} In particular, since  $\pm  \breve a(0)>-\varepsilon$ on $\vX^\pm_\varepsilon $, 
  also $\pm \breve a(t)>-\varepsilon$ on $\vX^\pm_\varepsilon $ for  $\pm t\geq 0$.  
   \end{proof}

For any  $n\in \N$ we construct $a^\pm_n=\sum_0^n
b^\pm_k$  as follows (omitting superscripts). Let $b_0=1$ and $q_0=q$.
Suppose  that  $b_k$ and $q_k$  are
 constructed for a given $ k\in\set{0,\dots, n-1}$,  then these quantities with $k$ replaced by $k+1$  are given by
\begin{align*}
  b_{k+1}&= \i\int _0^{\pm \infty} q_k(\Theta(t))\,\d t,\\
{q}_{k+1}&= qb_{k+1}-\tfrac 12 (\Delta_{(x,y)} b_{k+1}).
\end{align*} These  functions  
solve    transport equations, more precisely 
\begin{align}\label{eq:solut}
  \i\parb{\partial_\eta
+(\eta,\zeta) \cdot
\nabla_{(x,y)}}b_{k+1}=q_k=qb_{k}-\tfrac 12 (\Delta_{(x,y)} b_{k}).
\end{align}

 For the sake of justification  of the above recursion scheme we note the elementary computation
  \begin{align}
    \label{eq:elebnd}
    \begin{split}
    \forall f>0:
\quad \int _0^{ \infty} (t^2+f^2)^{-s_1} t^{s_2} &\,\d t=C_{s_1,s_2}
    \,f^{s_2+1-2s_1};\\&  s_2+1-2s_1<0,\quad-1<s_2.    
    \end{split}
\end{align}

 It   follows  from Lemma \ref{lemma:free-class-mech}, \eqref{eq:BEST}, \eqref{eq:elebnd}, the Fa\`a di Bruno formula  and  induction
 that for any $0\leq k\leq n$  and any   $\varepsilon\in(0,1)$
\begin{align}\label{eq:derBndq22i}
  \begin{split}
    &\big|\partial_{\eta,\zeta}^\alpha\,\partial_{x,y}^\beta b^\pm_k\big|\leq C_{\alpha,\beta}\parb{1+x+\inp{y}_m}^{-(k\delta +\abs{\alpha}/2+\abs{\beta})},\\
&\big| \partial_{\eta,\zeta}^\alpha\,\partial_{x,y}^\beta q^\pm_k\big|\leq C_{\alpha,\beta}\parb{1+x+\inp{y}_m}^{-(1/2+ (k+1)\delta+\abs{\alpha}/2+\abs{\beta})};\quad (x,y;\eta,\zeta)\in \vX^\pm_\varepsilon.   
  \end{split}
\end{align}
  In particular, 
\begin{align*}
\begin{split}
     \big|\partial_{\eta,\zeta}^\alpha\,\partial_{x,y}^\beta a^\pm_n\big|&\leq C_{\alpha,\beta}\parb{1+x+\inp{y}_m}^{-(\abs{\alpha}/2+\abs{\beta})},\\\big|\partial_{\eta,\zeta}^\alpha\,\partial_{x,y}^\beta
     q^\pm_n\big|&\leq C_{\alpha,\beta}\parb{1+x+\inp{y}_m}^{-(1/2+
         (n+1)\delta+\abs{\alpha}/2+\abs{\beta})}
;\quad (x,y;\eta,\zeta)\in \vX^\pm_\varepsilon.  
\end{split}
\end{align*} 
 Although we could work with $a^\pm_n$ for a fixed large $n$ it is
convenient to repeat the construction of the $b_k$'s without limit and
then 
invoke the Borel construction to regularize the sum $\sum_0^\infty
b^\pm_k$. 

Whence we
introduce for any $\varepsilon\in(0,1)$ the symbol
\begin{align}\label{eq:Bsym}
  a_B^\pm=\chi^\pm_\varepsilon \sum_0^\infty  \chi_k
b^\pm_k,\quad \chi^\pm_\varepsilon=\chi(\pm a>-\varepsilon),\quad \chi_k=\chi(f>C_k),
\end{align} for a suitable sequence $\sqrt 2 <C_0<C_1<\dots \nearrow
\infty$. Here and henceforth  we use  the abbreviated notation $
f=f_m$ and 
$a=a_m=\tfrac{\eta+\hat y_m\cdot\zeta}{f}$. As noted before $a=\breve
a$  for $f> \sqrt 2$. The relevant choice of $m$  depends on a
bounding constant of the 
energy $\lambda$, cf. a discussion in Section \ref{subsubsec:Resolvent bounds}, but
for convenience we prefer to suppress the  {dependence on} $m$
in our
notation. The construction of such 
sequence $(C_k)$  is standard (see  for example  the proof of \cite[Theorem
1.2.6]{Ho}), we provide the
details for our setting in  Appendix \ref{AppendixB}. Due to the fact that \eqref{eq:derBndq22i} are uniform
bounds it is not important for the construction that the
variable $\zeta$ is considered as  a bounded variable. However some   derivatives of the factor $\chi^\pm_\varepsilon
$ in \eqref{eq:Bsym} are only bounded locally  {uniformly} in  $\zeta$. 
Thus the notation $\vO(\cdot)$ below refers to a symbol 
obeying the
indicated 
  bound, however  this only  being locally  uniform in  $\zeta$. In
  addition we use the notation $\vO\parb{f^{-\infty}}$ to mean a
  smooth function (a  symbol) with all derivatives being of the form 
  $\vO\parb{f^{-k}}$ locally  uniform in  $\zeta$ (and uniform in the
  other variables) for any $k\in\N$.
In conclusion, thanks to \eqref{eq:solut},   there exists a suitable
sequence $(C_k)$  such that  for 
arbitrarily localized  $\zeta$  (and  for any fixed $m$)
\begin{align}\label{eq:Borel}
  \begin{split}
  \partial_{\eta,\zeta}^\alpha\,\partial_{x}^{\beta}\partial_{y}^{\gamma}
  a^\pm_B&=\vO\parbb{f^{-(\abs{\alpha}+2\abs{\beta})}\min\parb{f^2,\inp{y}_m}^{-\abs{\gamma}}},\\
  \i\parb{\partial_\eta
+(\eta,\zeta) \cdot
\nabla_{(x,y)}}a_B^\pm &=qa_B^\pm-\tfrac 12 (\Delta_{(x,y)}
  a_B^\pm)+\sum_0^\infty  r^\pm_k+\vO\parb{f^{-\infty}};
\\ r^\pm_k&=  \i {b}^\pm_k\, (\chi_k\partial_\eta\chi^\pm_\varepsilon
+(\eta,\zeta) \cdot \nabla_{(x,y)}(\chi_k\chi^\pm_\varepsilon ))\\& 
\quad \quad +(\nabla_{(x,y)}{b}^\pm_k) \cdot \nabla_{(x,y)}(\chi_k\chi^\pm_\varepsilon)+\tfrac{{b}^\pm_k}2\,\Delta_{(x,y)}(\chi_k\chi^\pm_\varepsilon).
  \end{split}
\end{align} These bounds in combination with
\eqref{eq:dereta} will  play a basic role in  the following Sections~\ref{subsubsec:Best results on the scattering
  matrix} and \ref{subsubsec:better Best results on the scattering
  matrix}.

\section{Analysis of  the scattering
  matrix} \label{subsubsec:Best results on the scattering
  matrix}
We
consider in this section double integrals of the form
\begin{align*}
  c\int \d
   \zeta \,\xi(\zeta)\int \e^{\i \theta_\lambda}\,\tilde a\,\d \eta;\\
c=(2\pi)^{-\tfrac{d+1}2},\quad \theta_\lambda&= y \cdot\zeta-\eta^3/6+(x+\lambda-\zeta^2/2)\eta.
\end{align*} Such  integrals were studied in Subsection
\ref{subsec:The stationary phase method} with symbols
$\tilde{a}=\tilde{a}(x,y;\eta,\zeta)$ obeying \eqref{eq:1symb} (see also the examples \eqref{eq:free_eigFct} and \eqref{eq:free_appr_eigFct}). The function $\xi$
can in some cases be considered as any  compactly supported distribution,
but of course the double integral has nicest properties for $\xi\in
C_\c^\infty(\R^{d-1})$ (as in Subsection
\ref{subsec:The stationary phase method}).

\begin{subequations}
Let $m\in\N$  (it is  considered as a large fixed 
auxillary parameter),  let $\varepsilon\in(0,1/2)$  (conveniently
taken small)
and let $a_B^\pm$  be the associated symbol  given by \eqref{eq:Bsym}.  We introduce  then the explicit example 
\begin{align}\label{eq:funFor0i}
\begin{split}
   F^\pm_{m,\varepsilon}(\lambda)^* \xi:=c\int \d
   \zeta \,\xi(\zeta)\int \e^{\i \theta_\lambda}a^\pm_B\,\d \eta.
\end{split}
\end{align}  

We calculate using \eqref{eq:dereta},
\eqref{eq:Borel}  and an
integration by parts 
\begin{align}\label{eq:funFori}
  \begin{split}
   (H-\lambda) F^\pm_{m,\varepsilon}(\lambda)^* \xi&=-c\int \d
  \zeta \,\xi(\zeta)\int \e^{\i
    \theta_\lambda}\parbb{ \vO\parb{f^{-\infty}} +\sum_0^\infty
    r^\pm_k }\,\d \eta; \\
  r^\pm_k&=  \i {b}^\pm_k\, \parb{\chi_k\partial_\eta\chi^\pm_\varepsilon
+(\eta,\zeta) \cdot \nabla_{(x,y)}(\chi_k\chi^\pm_\varepsilon)}\\& 
\quad \quad +(\nabla_{(x,y)}{b}^\pm_k) \cdot \nabla_{(x,y)}(\chi_k\chi^\pm_\varepsilon)+\tfrac{{b}^\pm_k}2\,\Delta_{(x,y)}(\chi_k\chi^\pm_\varepsilon).
  \end{split}
\end{align} We can  use \eqref{eq:derBndq22i} to estimate the
quantities of 
\eqref{eq:funFor0i} and \eqref{eq:funFori}, that is
\begin{align*}
   \phi_{m,\varepsilon}^{\lambda\pm}[\xi]:=
   F^\pm_{m,\varepsilon}(\lambda)^* \xi\quad \mand \quad
\psi_{m,\varepsilon}^{\lambda\pm}[\xi] :=(H-\lambda)\phi_{m,\varepsilon}^{\lambda\pm}[\xi].
 \end{align*} We 
 obtain  for the corresponding symbols, say denoted by $\tilde{a}_1$
 and   $\tilde{a}_2$ respectively, the   bounds
 \begin{align}\label{eq:modbnd}
   \begin{split}
   \abs{\partial_{\eta,\zeta}^\alpha&\,\partial_{x}^{\beta}\partial_{y}^{\gamma}
   \tilde{a}_1}\leq
   C_1f^{-(\abs{\alpha}+\abs{2\beta})}\min\parb{f^2,\inp{y}}^{-\abs{\gamma}};\quad
   C_1=C_1(\zeta),\\
   \abs{\partial_{\eta,\zeta}^\alpha\,\partial_{x,y}^{\beta}
   &\tilde{a}_2}\leq
   C_2\parbb{1+\tfrac{\abs{\eta}}f}\min\parbb{f,\sqrt{\inp{y}}}^{-(\abs{\alpha}+\abs{2\beta}+1)};\,\,
   C_2=C_2(\zeta).
  \end{split}
\end{align}  
 Here and henceforth  we use  the function $
f=f_m$ of Sections  \ref{subsubsec:Resolvent bounds} and \ref{subsubsec:Classical
  mechanics bounds},  and we use, slightly
abusively, the  notation $\inp{y}$ for $\inp{y}_m$. Note that the constants depend on $\zeta$, although if
$\zeta\in B_R=\set{\abs{\zeta}<R}$   for any given
 $R>0$, then the dependence is via $R$ only. Also there is a
 dependence on the multiindices, however  for convenience not
 indicated. Note that the first bound of \eqref{eq:modbnd} corresponds
 to the first assertion of \eqref{eq:Borel}. The second  bound of
 \eqref{eq:modbnd} follows readily from an 
 examination of the expressions $ r^\pm_k$ and the concrete
 construction given in Appendix \ref{AppendixB}.

\end{subequations}

We claim  the following formula for the generalized eigenfunction of
\eqref{eq:rep1}
\begin{subequations}
  \begin{align}\label{eq:rep12i}
    \begin{split}
    \phi_{\rm{ex}}^{\lambda \pm}[\xi]= \phi_{m,\varepsilon}^{\lambda \pm}[\xi]-R(\lambda\mp\i 0)\psi_{m,\varepsilon}^{\lambda \pm}[\xi];\quad\xi\in
    C_\c^\infty(\R^{d-1}),  
    \end{split}
  \end{align} and the related formulas  
\begin{align}\label{eq:0Sommer2i}
    0= \phi_{m,\varepsilon}^{\lambda \pm}[\xi]-R(\lambda\pm\i 0)\psi_{m,\varepsilon}^{\lambda \pm}[\xi],
  \end{align} leading to  
  \begin{align}\label{eq:rep32i}
    \xi= \pm\i 2\pi  \vF^\pm(\lambda)\psi_{m,\varepsilon}^{\lambda \pm}[\xi];\quad \xi\in
    C_\c^\infty(\R^{d-1}). 
  \end{align}  
\end{subequations}

The analogue of \eqref{eq:formS2} reads,
thanks to \eqref{eq:exForm}, \eqref{eq:rep12i} and \eqref{eq:rep32i}, 
\begin{align}\label{eq:formS22i}
  \begin{split}
    \tfrac 1{2\pi\i}\inp{{\xi},S(\lambda)\xi'}=\inp{\psi_{m,\varepsilon}^{\lambda+}[{\xi}],R(\lambda+\i
      0)\psi_{m,\varepsilon}^{\lambda-}[\xi']}-\inp{\phi_{m,\varepsilon}^{\lambda+}[{\xi}],\psi_{m,\varepsilon}^{\lambda-}[\xi']}.
\end{split}
\end{align} 

Note that the leading order asymptotics of
$\phi_{m,\varepsilon}^{\lambda \pm}[{\xi}]$ for $\xi\in
    C_\c^\infty(\R^{d-1})$ follows from
\eqref{eq:basyp}. The expansion  terms from
stationary
phase analysis, cf. Appendix \ref{Appendix},   all  vanish for
$\psi_{m,\varepsilon}^{\lambda \pm}[{\xi}]$. The following result is a
manifestation of this fact.

\begin{lemma}
  \label{lem:eigFori} For all $\xi\in C_\c^\infty(\R^{d-1})$ the
  functions $\psi_{m,\varepsilon}^{\lambda \pm}[{\xi}]\in L^2_\infty$, and
  the formulas \eqref{eq:rep12i} and \eqref{eq:0Sommer2i} are valid.
\end{lemma}
\begin{proof}
  Write
\begin{align*}
    \psi_{m,\varepsilon}^{\lambda \pm}[\xi]
    =\psi_{m,\varepsilon,1}^{\lambda \pm}[\xi]+\psi_{m,\varepsilon,2}^{\lambda \pm}[\xi].
  \end{align*} 
   corresponding to the splitting
  \begin{align*}
    \int \e^{\i
    \theta_\lambda}\parbb{ \vO\parb{f^{-\infty}} +\sum_0^\infty
    r^\pm_k }\,\d \eta = \int \e^{\i
    \theta_\lambda}\,\vO\parb{f^{-\infty}}\,\d \eta+\int \e^{\i
    \theta_\lambda}\,\sum_0^\infty
    r^\pm_k \,\d \eta
  \end{align*} in \eqref{eq:funFori}.
There are two ways of integrating by
  parts using
\eqref{eq:inteparts1i} and \eqref{eq:inteparts2i}, respectively.

  \Step{I} The contribution from $\vO\parb{f^{-\infty}}$ takes the
  desired  form  since we have any (high) power 
$(x+\inp{y})^{-t}= 2^{t}f^{-2t}$ at our disposal. We can use
a part of this factor to obtain a  high power of $\inp{
  (\eta,\zeta)}^{-1}$  as well as  a  high power of  $\inp{x}^{-1}$  by integrating by parts using
\eqref{eq:inteparts1i} repeatedly.  We can then use
the decay in $x$ (needed for $x<0$ only) and another part of the  factor
$(x+\inp{y})^{-t}$ (for $x\in \R$ arbitrary) to obtain  a
high power of $\inp{y}^{-1}$ as well.
 Altogether we obtain a  desirable   factor
$\inp{x}^{-s}\inp{y}^{-s}$ with $s>1$ arbitrarily big for the contribution
from $\vO\parb{f^{-\infty}}$.

 \Step{II} As for the  contribution from the terms $r^\pm_k $ of $\sum_0^\infty
    r^\pm_k$ we observe that terms for
which the
factor $\chi_k$ is differentiated  can be treated exactly as above. 

\Step{III} It
remains to consider the contributions from terms where at least one
derivative falls on the factor $\chi_\varepsilon^\pm$. By definition any such term
is supported in $\set{x+\inp{y}>1,\, -\varepsilon/2\geq \pm a\geq
  -\varepsilon}$. We mimic Subsection \ref{subsec:The stationary phase
  method} and Appendix 
\ref{Appendix}. Since $\xi$ is compactly supported the variable
$\zeta$ is localized and we may  for any such term consider
\begin{align}\label{eq:appro}
 \tfrac{|\eta|}{\sqrt{2x+2\inp{y}}}\approx  |a|\in [\varepsilon/2, \varepsilon]\quad \text{effectively}.
\end{align}

 Since $\varepsilon \in (0,1/2)$ the stationary points
 \eqref{eq:asypFixed} do not conform with  \eqref{eq:appro}.  This
 means that Remark \ref{remark:statproof-eqrefeq:asysta} applies,
 proving   the first  assertion of the lemma.

\Step{IV}  As for the second assertion, the difference of  the
left- and right-hand sides in \eqref{eq:rep12i} is a purely  incoming  or outgoing
generalized eigenfunction in $\vB^*$, respectively,
cf. \eqref{eq:basyp} and  the proof
of  
\cite[Lemma 4.4]{AIIS2}. Hence by Theorem
\ref{thm:char-gener-eigenf-1} \ref{item:14.5.13.5.40}  it 
vanishes. For \eqref{eq:0Sommer2i} we can argue  similarly.
\end{proof}

 Let $\vE'_{d-1}$ denote the  space of compactly supported
distributions on $\R^{d-1}$. Any   $\xi\in\vE'_{d-1}$ can be written as
$\xi=Q\xi'$,  where $Q=Q(\zeta,p_\zeta)$ is a differential operator on $\R^{d-1}$ with
smooth coefficients  and $\xi'\in
C_\c(\R^{d-1})$. We shall use the quantity $A=A_m$  of Section \ref{subsubsec:Resolvent bounds}.
Let for any $\varepsilon\in(0,1/2)$
\begin{align}\label{eq:locEpsi}
  \chi_{-\varepsilon}(t)=\chi(t<-\varepsilon/4)\chi(t>-2\varepsilon);\quad
  t\in\R.
\end{align}  
\begin{lemma}\label{lemma:best-results-scatt} Let $n\in \N_0$ and
  consider a fixed $\xi\in\vE'_{d-1}$  {of the form}
  $\xi=\inp{p_\zeta}^{2n}\xi'$, $\xi'\in
C_\c(\R^{d-1})$. Then    there exists  $s'=s'(n)\in\R$ such that for any
 $s\in\R$,  the
  quantities 
 $\psi_{m,\varepsilon}^{\lambda\pm}[\xi]$ are   represented  
  \begin{align}\label{eq:repG}
   \psi_{m,\varepsilon}^{\lambda\pm}[\xi] = f^{s'}\chi_{-\varepsilon}(\pm
    A)\varphi_1^\pm+f^{-s}\varphi_2^\pm\text { for some  }\varphi_1^\pm,\,\varphi_2^\pm\in \vH. 
  \end{align}
   The $\vH$-norm of 
 $\varphi_1^\pm$ and  $\varphi_2^\pm$  
 can be 
estimated by  $CR^{(d-1)/2}\norm{\xi'}_\Sigma$ provided $\supp \xi' \subseteq B_R$,
 and $\varphi_1^\pm$ (and  the corresponding constant $C$) can be chosen independent of $s$.
  \end{lemma}
  \begin{proof} \Step{I} 
    The powers of $p_\zeta$ in the definition of the $\xi$ can be moved to
    other factors  of the $\zeta$-integral thereby 
    producing  additional  factors of monomials in  $y-\eta\zeta$. By the proof of
    Lemma \ref{lem:eigFori} the contribution from the term $\vO\parb{f^{-\infty}}$
    did not use integration by parts in the other direction,
    i.e. \eqref{eq:inteparts2i} was not used. In fact an arbitrarily  large negative
    power  $\inp{x}^{-s}\inp{y}^{-s}$  was
    produced. This can bound a  factor $\inp{y}^{2n}$, and we
    conclude  that the contribution from the term
    $\vO\parb{f^{-\infty}}$  takes the form
    of  the second term on the right-hand side  of \eqref{eq:repG}.

\Step{II}  As for the  contribution from the $r^\pm_k $'s we observe that terms for
which the
factor $\chi_k$ is differentiated  offer a factor $f^{-s}$ right away
(in fact for any $s$) and we can also bound  an additional   factor $\inp{y}^{2n}$, so again there is agreement with the form of
the second term on the right-hand side  of \eqref{eq:repG}.

\Step{III}  As for the  contribution from the terms of $r^\pm_k $ for
which the
factor $\chi^\pm_\varepsilon$ is differentiated  is more
complicated. We can not proceed as in  Step III of the proof of Lemma
\ref{lem:eigFori}, since now 
 integration by parts using  \eqref{eq:inteparts2i}  is not doable. In a region of
the form
\begin{align}\label{eq:loc}
  \set{x>1,\, \abs{\abs{\eta}- \sqrt{2x}}>\epsilon}\cup \set{x<4};\quad
  \epsilon>0,
\end{align} we obtain  a high power of
$\inp{(x,\eta)}^{-1}$ by the $\eta$-integration by parts. This power
in combination with the growing factor $\inp{y}^{s'/2}$, $s'/2=2n+d$, can be bounded by
$f^{s'}$. This leads us to writing  the  contribution from
any  term
given by first 
localizing to \eqref{eq:loc}
  as $f^{s'}\varphi^\pm$ with $\varphi^\pm\in \vH$, and therefore  in turn as
    \begin{align}\label{eq:deGo}
      f^{s'}\varphi^\pm=f^{s'}\chi_{-\varepsilon}(\pm
        A)\varphi_1^\pm+f^{s'}\parb{1-\chi_{-\varepsilon}(\pm
      A)}\varphi^\pm \text { with }\varphi_1^\pm=\varphi^\pm\in \vH.
    \end{align}
 The first term agrees with the first term on the right-hand side  of
 \eqref{eq:repG}, so it remains to  show that the second term agrees
 with the second  term on the right-hand side  of
 \eqref{eq:repG}. For the latter task we observe that if we replace
 $A$ by its Weyl symbol  $a_{\rm W}$ ($=a=\tfrac{\eta+\hat y_m\cdot
    \zeta}{f}$ for ${x+\inp{y}}>1$),
  then at this rough symbolic level obviously 
\begin{align}\label{eq:sym}
  \parb{1-\chi_{-\varepsilon}(\pm a_{\rm W})}f^{-s'}\chi_k\chi'(\pm a>-\varepsilon)=0.
\end{align} We are discussing  the case where  $\chi^\pm_\varepsilon$
is differentiated and the prime for the third  factor $\chi$ denotes the
derivative of the function. Terms  with the double derivative
$\chi''(\pm a>-\varepsilon)$ can be treated similarly as below.
We may  move the factor $1-\chi_\varepsilon(\pm A)$ inside the integrals pass the exponential $\e^{\i
    \theta_\lambda}$ and then replace the operator by its
  symbol (which should be legitimate to leading order) and finally
  conclude by \eqref{eq:sym}. However there are `errors'
  due to $(x,y)$-dependence of the given  symbols. We  implement
  a version of this scheme below. 

Pick  $\chi\in C^\infty_\c(\R)$ with $
  \chi(t)=1 $ on $\supp \chi'\parb{\cdot>-\varepsilon}$ but $
  \chi(t)=0 $ on $\supp (1-\chi_{-\varepsilon})$.
 Take an almost analytic extension $\tilde \chi\in C_\c^{\infty }(\C)$ of $\chi$,
and set 
\begin{align*}
\mathrm d\mu_\chi(z)=\pi^{-1}(\bar\partial\tilde \chi)(z)\,\mathrm du\mathrm dv;\quad 
z=u+\i v.
\end{align*}
Then 
\begin{align*}
\chi(t) 
=
\int _{\C}(t -z)^{-1}\,\mathrm d\mu_\chi(z);\quad t\in\R.
\end{align*} In particular
\begin{align*} &\parb{1-\chi_{-\varepsilon}(\pm A)}f^{-s'}\chi_k\chi'(\pm a>-\varepsilon)\\
  &=\parb{1-\chi_{-\varepsilon}(\pm A)}\parb{\chi(\pm a)-\chi(\pm
    A)}f^{-s'}\chi_k\chi'(\pm a>-\varepsilon)\\
&=\pm\int _{\C}\parb{1-\chi_{-\varepsilon}(\pm A)}(\pm A-z)^{-1}\parb{(A-a)f^{-s'}\chi_k}(\pm a-z)^{-1}\chi'(\pm a>-\varepsilon)\,\mathrm d\mu_\chi(z).
\end{align*}

We insert this formula in the expression for $\varphi^\pm$ for those
terms  with a single derivative of  $\chi(\pm \cdot>-\varepsilon)$  (the one with a double
derivative can be treated similarly). Then we move the middle factor
$(A-a)f^{-s'}\chi_k$ to the far right, in particular pass  the exponential $\e^{\i
    \theta_\lambda}$. This produces altogether an extra factor
  $f^{-1}$ since all derivatives (i.e. components of $p$ applied to
  functions) are  bounded except when passing
  through the exponential where  a cancellation occurs. Repeating this procedure we  gain a large power of
  $f^{-1}$, in particular a factor  $f^{-s-s'}$, which 
  allows us to  conclude that  the
  second term of \eqref{eq:deGo} agrees
 with the second  term on the right-hand side  of
 \eqref{eq:repG}.

\Step{IV}  For a localized term in the region of the form
$\set{x>2,\, \abs{\abs{\eta}- \sqrt{2x}}<2\epsilon}$,
which remains to be treated, the 
$\eta$-integration by parts in the beginning   of Step III does not work. To get a weight like $\inp{x}^{-j}$
 (and therefore $\inp{(x,\eta)}^{-j}$)  to insure the Hilbert space bound, we simply  bound
the   {$j$-th} power of $\inp{x}$ by the same power of $f^2$, yielding the desired
inverse power of $\inp{x}$. Here $j=3+n$ suffices, and with the next
argument of  Step III we conclude \eqref{eq:deGo} with
$s'=2j+4n+2d$. Then we  
  mimic the last   part of Step III.

\Step{V} Clearly the $\varphi_1^\pm$ resulting from the combination of
Steps III and IV  above is {independent of} $s$, and our arguments lead
in all cases to  $\vH$-norm bounds with a {dependence on} $\xi'$ only
through a factor of 
$\int\,\abs{\xi'}\,\d \zeta$,  and therefore  by \cas in turn through a factor of 
 $R^{(d-1)/2}\norm{\xi'}_\Sigma$.
\end{proof}

\begin{corollary}\label{cor:impr-form-scatt} For all $\xi\in
  C_\c^\infty(\R^{d-1})$ the vector  $S(\lambda)\xi\in
  C^\infty(\R^{d-1})$. 
  \end{corollary}
  \begin{proof} Let $\xi=\xi_-\in
  C_\c^\infty(\R^{d-1})$ be given. Let
  $\xi_+=\inp{p_\zeta}^{2n}\xi'_+$ be given as in  Lemma \ref{lemma:best-results-scatt}.
\Step{I}  We look at the first term
  $\inp{\psi_{m,\varepsilon}^{\lambda+}[{\xi_+}],R(\lambda+\i
    0)\psi_{m,\varepsilon}^{\lambda-}[\xi_-]}$ of
  \eqref{eq:formS22i}. By combining the 
  representations of $\psi_{m,\varepsilon}^{\lambda\pm}[{\xi_+}]$ from  Lemma \ref{lemma:best-results-scatt} 
   in combination with
  \eqref{eq:LAPi}--\eqref{eq:twosidedbi} it follows  
         by the  Sobolev embedding theorem \cite[4.5.13]{Ho}  that
         indeed the contribution {to} $S(\lambda)\xi_-$ from the first
term of \eqref{eq:formS22i} is smooth.  This argument 
only uses  the weak
  input  $\xi_-\in
C_\c(\R^{d-1})$. More generally it works with  also
$\xi_-=\inp{p_\zeta}^{2n}\xi'_-$  given as in  Lemma
\ref{lemma:best-results-scatt} for an arbitrary   $n\in\N_0$.
\Step{II}
The arbitrary power  decay of $\psi_{m,\varepsilon}^{\lambda-}[\xi_-]$  from   Lemma
\ref{lem:eigFori}  and repeated integration by parts (using
\eqref{eq:inteparts1i})  in the integral
$\phi_{m,\varepsilon}^{\lambda+}[\xi_+]$  yield that also the second  term of
\eqref{eq:formS22i} contributes by a  smooth term to
$S(\lambda)\xi_-$. Note that we can estimate  the term by 
$C\norm{\xi'_+}_\Sigma$  and consequently  again 
invoke the  Sobolev embedding theorem.
\end{proof}

We noted  in Step I in the above proof that the smoothness of $\xi_-(\in
  C_\c^\infty(\R^{d-1})$ was  not used for that part of the proof.
   We could have assumed $\xi_-\in\vE'_{d-1}$
  only and concluded that the corresponding contribution to $ S(\lambda)\xi_-\in
  C^{\infty}(\R^{d-1})$. In fact it follows readily (for example by
  using the left Kohn-Nirenberg
quantization  discussed below) that the  first term
   of \eqref{eq:formS22i} is represented by a \emph{smoothing operator} in the sense
   of \eqref{eq:orderDef}. Whence  the
   local singularities of the kernel $S(\lambda)(\zeta,\zeta')$ appear
   in the   second  term
   of \eqref{eq:formS22i} only. This term is an \emph{explicit oscillatory
   integral}. We have
   \begin{align*}
     -2\pi \i \inp{\phi_{m,\varepsilon}^{\lambda+}[{\xi}],\psi_{m,\varepsilon}^{\lambda-}[\xi']}&\\=
\tfrac{-\i}{(2\pi)^{d}}\int\!\!\!\!\int\d x\d y \int  \d
  \zeta \overline{\xi(\zeta)}\int \e^{-\i
    \theta_\lambda}\overline{a^+_B }\d \eta& \int \d
  \zeta' \xi'(\zeta')\int \e^{\i
    \theta'_\lambda}\,\parbb{ \vO\parb{f^{-\infty}} +\sum_0^\infty
    r^-_k }    \d \eta',
   \end{align*} where the first exponential $\e^{-\i
    \theta_\lambda}$ is considered as a funtion of $(\eta,\zeta)$
  (and of $(x,y)$ as well) while  the second  exponential $\e^{\i
    \theta'_\lambda}=\e^{\i
    \theta_\lambda}$ is considered as a function of
  $(\eta',\zeta')$. Of course the
  symbols $\vO\parb{f^{-\infty}}$ and the $r^-_k$'s  also depend of the
  variables $(\eta',\zeta')$, while $a_B^+$
  rather 
depends on $(\eta,\zeta)$. Up to a convergence factor
$\chi(\abs{y}/R<1)$ (with $R\to \infty$) we write the right-hand side  as
\begin{align*}
 \int\!\!\!\!\int\,\overline{\xi(\zeta)}\breve S(\zeta,\zeta')
{\xi'(\zeta')}\,\d \zeta\d \zeta' 
\end{align*} and then in turn $\breve S$ as a pseudodifferential
operator with corresponding symbol $\breve s$
\begin{align*}
  \breve S(\zeta,\zeta') &=(2\pi)^{1-d}\int\, \e^{\i(\zeta-\zeta')\cdot
  y} \,\breve s(\zeta,\zeta',-y) \,\d y;\\
\breve s(\zeta,\zeta',y)&=( 2\pi \i)^{-1}\int\d
                   x\int {\e^{-\i\varphi_\lambda(x, \eta,\zeta)}\,\overline{a^+_B }}\,\d
                   \eta\\
&\quad \quad
  \quad\quad\int{\e^{\i\varphi_\lambda(x, \eta',\zeta') }}\parbb{ \vO\parb{f^{-\infty}} +\sum_0^\infty
    r^-_k } \,\d
  \eta';\\&\quad \varphi_\lambda(x, \eta,\zeta)={-\eta^3/6+(x+\lambda-\zeta^2/2)\eta}.
\end{align*} 

\begin{thm}
  \label{thm:analys-scatt-matra} The scattering operator $S(\lambda)$
  is a PsDO of order $0$ in the sense of \eqref{eq:orderDef}.
\end{thm}

 We are going to use 
Corollary \ref{cor:impr-form-scatt} (and  Step I in  its proof) to establish the theorem.  For a \emph{suitable}
realization of a  symbol $s$ of
$S(\lambda)$ we need to verify  the following bounds.
\begin{align}\label{eq:derRE0}
  \begin{split}
   &\forall \alpha,\alpha', \beta\in \N_0^{d-1}:\\ \abs{ \partial^{\alpha}_{\zeta}\partial^{\alpha'}_{\zeta'}\partial^\beta_y
  s}&\leq
C_{\alpha,\alpha',\beta}\,\inp{y}^{-\abs{\beta}}\text{
  for all } y \text{ and
  locally  uniformly in }\zeta,\zeta'.  
  \end{split} 
\end{align} Due  to Step
I of the  proof of Corollary \ref{cor:impr-form-scatt} and the
subsequent discussion only  $\breve S$ and a corresponding symbol
$\breve s$ need  elaboration. Our   realization of $\breve s$ is
conveniently given by the  left Kohn-Nirenberg
symbol $s_{\rm KN}$, implicitly given by
\begin{align}\label{eq:KNf}
   \breve S(\zeta,\zeta') &=(2\pi)^{1-d}\int\, \e^{\i(\zeta-\zeta')\cdot
  y} \,\breve s_{\rm KN}(\zeta,y) \,\d y.
\end{align}
\begin{proof}[Proof of Theorem \ref{thm:analys-scatt-matra}] 
  \Step{I}   The
contribution from $\vO\parb{f^{-\infty}}$, say with PsDO symbol $s^-$,  is a smoothing operator,
since we can use the bound {$(x+\inp{y})^{-s}$,} $s$ large, and
$\eta$- and $\eta'$-integration by parts to obtain the following  bound of any
derivative,
\begin{align*}
  \partial^\alpha_{\zeta,\zeta',y} s^-(\zeta,\zeta',y)=\vO\parb{\inp{(\zeta,\zeta',y)}^{-\infty}}.
\end{align*}
\Step{II} We are left with examining the symbol
\begin{align}\label{eq:simp}
  t=( 2\pi \i)^{-1}\int\d x\int
  \e^{-\i\varphi_\lambda(x, \eta,\zeta)} \, {\overline{a_B^+}}\,\d
  \eta
\int{\e^{\i\varphi_\lambda(x, \eta',\zeta')} }\,\sum_0^\infty
    r^-_k\,\d \eta'.
\end{align} We shall here bound this expression locally uniformly in
$(\zeta,\zeta')$ and unifomly in $y$.

 By the non-stationary phase argument (in the variable
 $(x,\eta,\eta')$) we obtain,  computing up to a  smoothing operator, that only a
localization to 
$\set{x>R,\,\eta\approx \eta',\,{\eta}\approx \sqrt{2x}}$ with $R>1$ big 
matters. 
In turn with such localization we can use the method of stationary
phase, cf.   Subsection
\ref{subsec:The stationary phase method} and Appendix
 \ref{Appendix}. With reference to the
notation of  \eqref{eq:mainP} and 
\eqref{eq:locEpsi} it suffices  to bound the following expression
\begin{align*}
  ( 2\pi \i)^{-1}\int\d x\,\chi_R(x)&\int
  \e^{-\i\varphi_\lambda(x, \eta,\zeta)} \, {\overline{a_B^+}}\,\chi_\epsilon\parb{{\eta}-\sqrt{2x}}\d
  \eta\\
&\int{\e^{\i\varphi_\lambda(x, \eta',\zeta')} }\,\chi_\epsilon\parb{{\eta'}-\sqrt{2x}}\,\chi_{-\varepsilon}\parbb{-\tfrac{\eta'}{\sqrt{2x+2\inp{y}}}}\sum_0^\infty
    r^-_k\,\d \eta'.
\end{align*} Note that the functions  $\chi_\epsilon$ and
$\chi_{-\varepsilon}$ are very different: By definition
$\chi_\epsilon$ is supported in $(-2\epsilon,2\epsilon)$  while $\chi_{-\varepsilon}$ is
supported in $(-2\varepsilon,-\varepsilon/4)$. In particular we can for free
insert  the localization factor
\begin{align*}
  \chi(x,y;\epsilon):=\chi\parb{{\sqrt{2x}}/{\sqrt{2\inp{y}}}<5\varepsilon}\chi\parb{{\sqrt{2x}}/{\sqrt{2\inp{y}}}>\varepsilon/9}
\end{align*} 
 provided  $\epsilon,\varepsilon>0$ are chosen sufficiently small
(which we can assume). 

We introduce the `large parameter'  $h^{-1}=\sqrt{2\inp{y}}$,  make the change of variable  $x\to h^{-2}x$ and  write
 \begin{align*}
  \varphi_\lambda\parb{=\varphi_\lambda(x,
   \eta,\zeta)}=h^{-1}\,\tilde{\varphi}_\lambda\,\mand \,\varphi'_\lambda\parb{=\varphi_\lambda(x,
   \eta',\zeta')}=h^{-1}\,\tilde{\varphi}'_\lambda.
 \end{align*} The integration in the new $x$ is over a
 compact interval due to the introduced factor $\chi\parb{\sqrt{2x}<5\varepsilon}\chi\parb{\sqrt{2x}>\varepsilon/9}$ and the double $(\eta,\eta')$-integral  can be
 estimated by the stationary phase method, cf. Appendix
 \ref{Appendix}. We skip the details of proof  at this point noting that one may
 mimic Appendix
 \ref{Appendix} interchanging  the roles of $x$ and $y$. This  leads to the bound $ h^{-2}\vO(h) \vO\parbb{\sum_0^\infty
    r^-_k}=\vO(h^{ 0})$, i.e.  uniform boundedness
 holds 
 (locally only  in
$(\zeta,\zeta')$) as desired.

\Step{III} It remains to bound derivatives of the above symbol $t$ of
\eqref{eq:simp}. Here we show the following weaker bounds.
\begin{align}\label{eq:derRE}
  \begin{split}
   &\forall \alpha,\alpha', \beta\in \N_0^{d-1}:\\ \abs{ \partial^{\alpha}_{\zeta}\partial^{\alpha'}_{\zeta'}\partial^\beta_y
  t}&\leq
C_{\alpha,\alpha',\beta}\,\inp{y}^{{\abs{\alpha}}/2+{\abs{\alpha'}}/2-\abs{\beta}}\text{
  for all } y \text{ and
  locally  uniformly in }\zeta,\zeta'.  
  \end{split}
\end{align}
 For that we first compute the derivatives by differentiating inside
 the integrals, then we mimic Step II invoking the bounds
 \eqref{eq:modbnd}. Since the phases $\varphi_\lambda$ and
 $\varphi'_\lambda$ are independent of $y$,  derivatives in $y$ conform
 with \eqref{eq:derRE}. However derivatives in $\zeta$ and $\zeta'$ are
 not so good as \eqref{eq:modbnd} indicates. The reason is that  the phases $\varphi_\lambda$ and
 $\varphi'_\lambda$ have a {dependence on} these variables and we can only
 bound like $\partial^{\alpha}_{\zeta}\varphi_\lambda =\vO(\inp{\eta})$ and
 $\partial^{\alpha}_{\zeta'}\varphi'_\lambda =\vO(\inp{\eta'})$ for
 $\abs{\alpha}\geq 1$. Effectively, cf. Step II,
 $\vO(\inp{\eta})=\vO(\inp{y}^{1/2})$ and
 $\vO(\inp{\eta'})=\vO(\inp{y}^{1/2})$. Whence \eqref{eq:derRE}
 follows.

\Step{IV} The left Kohn-Nirenberg symbol $t_{\rm KN}$, cf. \eqref{eq:KNf}, 
is obtained from $t$ by the formula
\begin{align*}
  t_{\rm KN}(\zeta,y)=\e^{\i p_{\zeta'}\cdot p_y} t(\zeta,\zeta',y)_{|\zeta'=\zeta},
\end{align*} cf. \cite[Theorems 18.4.10 and 18.5.10]{Ho2}. (Note that
this symbol is conveniently represented in momentum space,
cf. Appendix \ref{Appendix}.)  Whence (formally)
\begin{align*}
  t_{\rm KN}(\zeta,y)&=( 2\pi \i)^{-1}\int\d
                   x\int \d
                   \eta\,\e^{-\i\varphi_\lambda(x, \eta,\zeta)}\,\\
&\quad \quad
  \quad\quad\e^{\i p_{\zeta'}\cdot p_y}\,\parbb{{\overline{a_B^+}}\int{\e^{\i\varphi_\lambda(x, \eta',\zeta') }}\parbb{ \vO\parb{f^{-\infty}} +\sum_0^\infty
    r^-_k }   \,\d
  \eta'}_{|\zeta'=\zeta}.
\end{align*} However we only need the formula with the cutoffs as in
Step II. The point is that, when expanding the exponential, although
when $p_{\zeta'}$ hits the factor $\e^{\i\varphi_\lambda(x,
  \eta',\zeta') }$ a `growing' factor $\eta'$ is introduced which
effectively counts for a factor $f$, but the accompanying factor $p_y$
effectively counts for a factor $f^{-2}$. This means that the terms in
the expansion of $\e^{\i p_{\zeta'}\cdot p_y}$ effectively have
decreasing order. Truncating  the series leads to a more complicated
integrand, however the variable $\zeta'$ has disappeared and now the
phase factors enter as
\begin{align*}
  \e^{-\i\varphi_\lambda(x,
  \eta,\zeta) }\,\e^{\i\varphi_\lambda(x,
  \eta',\zeta) }.
\end{align*} The $\zeta$-derivatives applied to
this product are accompanied by factors of powers of $\eta-\eta'$
which can be written as  powers of $p_x$ applied to the same
product. Integrating by parts  in $x$ then effectively gives  inverse
powers of $f$. In particular the symbol does not become any worse when
differentiating with respect to  $\zeta$, {as we wish}. We can now
improve \eqref{eq:derRE},  
using  the stationary phase method
of  Step  II to  bound derivatives as in \eqref{eq:derRE0}.
\end{proof}

\begin{remark*} A  closer examination of  Step IV above shows the
  slightly stronger assertion on the symbol $s=s_{\rm KN}$:
  \begin{align}\label{eq:derRE02}
  \begin{split}
   &\forall \alpha, \beta\in \N_0^{d-1}:\\ \abs{ \partial^{\alpha}_{\zeta}\partial^\beta_y
  s}&\leq
C_{\alpha,\beta}\,\inp{y}^{-\abs{\alpha}/2-\abs{\beta}}\text{
  for all } y \text{ and
  locally  uniformly in }\zeta.  
  \end{split} 
\end{align}
\end{remark*}

From the fact that $ S(\lambda)$ is a  PsDO we deduce   the following
result (by a general
argument).

\begin{corollary}\label{cor:impr-form-scatt2i} The kernel  $S(\lambda)(\zeta,
  \zeta')$ is smooth away from the diagonal $\set{\zeta=
  \zeta'}$. 
  \end{corollary}

\section{The kernel of the scattering
  matrix at the diagonal} \label{subsubsec:better Best results on the scattering
  matrix}

We will derive yet another representation of the scattering
matrix. This will be more suitable for  computing
singularities at the diagonal of its kernel. We do  an 
analysis of the latter problem in Subsection \ref{subsubsec:Tbetter Best results on the scattering
  matrix2}.

\subsection{Subtracting the $\delta$-singularity at the diagonal} \label{subsubsec:Subtracting the delta-singularity at the diagonal}

Let $m$ and $\varepsilon$ be as in Section \ref{subsubsec:Best results on the scattering
  matrix}. We shall use notation of \eqref{eq:free_appr_eigFct},
\eqref{eq:derBndq22i}, \eqref{eq:Bsym}  and  the quantity
\begin{align}\label{eq:funFor0i3}
\begin{split}
   F^\pm_{m,\varepsilon}(\lambda)^* \xi&:=c\int \d
  \zeta \,\xi(\zeta)\int \e^{\i \theta_\lambda}\,(a^\pm_B+\chi_1 \chi^{\perp\pm}_\varepsilon)\,\d \eta;\\
&\chi^\pm_\varepsilon=\chi\parb{\pm
  a>-\varepsilon},\\&\chi^{\perp \pm}_\varepsilon=\chi^\perp\parb{\pm a>-\varepsilon};\quad a=\tfrac{\eta+\hat y_m\cdot
    \zeta}{f}.
\end{split}
\end{align}  Recall that $\chi_1=\chi(x+\inp{y}>C_1)$ and $\chi^\pm_\varepsilon=\chi\parb{\pm
  a>-\varepsilon}$ enter  in the definition \eqref{eq:Bsym}. Note in comparison with \eqref{eq:funFor0i}  the
appearence of the term  $\chi_1 \chi^{\perp\pm}_2$ (although the left-hand side notation for
convenience is the same).

We calculate
\begin{align}\label{eq:funFori3}
  \begin{split}
   (H-\lambda) F^\pm_{m,\varepsilon}(\lambda)^* \xi&=-c\int \d
  \zeta \,\xi(\zeta)\int \e^{\i
    \theta_\lambda}\parbb{\vO\parb{f^{-\infty}} +\sum_0^\infty
    r^\pm_k+r^{\perp \pm}_\varepsilon }\,\d \eta; \\ 
r^\pm_k&=  \i {b}^\pm_k\, \parb{\chi_k\partial_\eta\chi^\pm_\varepsilon
+(\eta,\zeta) \cdot \nabla_{(x,y)}(\chi_k\chi^\pm_\varepsilon)}\\& 
\quad \quad +(\nabla_{(x,y)}{b}^\pm_k) \cdot \nabla_{(x,y)}(\chi_k\chi^\pm_\varepsilon)+\tfrac{{b}^\pm_k}2\,\Delta_{(x,y)}(\chi_k\chi^\pm_\varepsilon),
\\ r^{\perp \pm}_\varepsilon=
\i \, \parb{\chi_1\partial_\eta\chi^{\perp \pm}_\varepsilon&
+(\eta,\zeta) \cdot
\nabla_{(x,y)}(\chi_1\chi^{\perp \pm}_\varepsilon)}+\tfrac{1}2\,\Delta_{(x,y)}(\chi_1\chi^{\perp \pm}_\varepsilon) -q \chi_1\chi^{\perp \pm}_\varepsilon.
  \end{split}
\end{align}

These formulas simplify as
\begin{align*} a^{\pm}_B +\chi_1 \chi^{\perp\pm}_\varepsilon&=1+\sum_1^\infty
\chi_k\chi^\pm_\varepsilon b^\pm_k+\vO\parb{f^{-\infty}},\\
r^\pm_B:=\sum_0^\infty
    r^\pm_k+r^{\perp \pm}_\varepsilon 
&= -q \chi_1\chi^{\perp\pm}_\varepsilon+\sum_1^\infty
    r^\pm_k +\parbb{1+\tfrac{\abs{\eta}}f}\vO\parb{f^{-\infty}}.
\end{align*}
  We compute, using \eqref{eq:derBndq22i}  and  \eqref{eq:B1}--\eqref{eq:B3},
 \begin{align}\label{eq:leadT}
   \begin{split}
a^{\pm}_B &=1-\chi_1 \chi^{\perp\pm}_\varepsilon+\vO\parb{f^{-{2\delta}}}=\vO\parb{f^{0}},\\
  r^\pm_B
  &=\parbb{1+\tfrac{\abs{\eta}}f}\vO\parbb{f^{-{2\delta}}\min\parb{f^2,{\inp{y}}}^{-1/2}},\\
   \breve r^\pm_B:&=r^\pm_B - \i {b}^\pm_1\, \parb{\chi_1\partial_\eta\chi^\pm_\varepsilon
+(\eta,\zeta) \cdot
   \nabla_{(x,y)}(\chi_1\chi^\pm_\varepsilon)}+q
 \chi_1\chi^{\perp\pm}_\varepsilon\\=\parbb{1+&\tfrac{\abs{\eta}}f}\vO\parbb{f^{-{4\delta}}\min\parb{f^2,{\inp{y}}}^{-1/2}}
 +\vO\parbb{f^{-{2\delta}}\min\parb{f^2,{\inp{y}}}^{-2}}.
   \end{split}
 \end{align} We may also compute derivatives and conclude that  $r^\pm_B$ fulfill the second bound of
 \eqref{eq:modbnd}  with an additional  factor
 $f^{-2\delta}$. We may also  derive similar (slightly stronger) bounds for
 derivatives of $\breve r^\pm_B$.

 Parallel to Section \ref{subsubsec:Best results on the scattering
  matrix}  we introduce  
 \begin{align*}
   \phi_{m,\varepsilon}^{\lambda\pm}[\xi]=
                                          F^\pm_{m,\varepsilon}(\lambda)^*
                                          \xi\quad \mand \quad
\psi_{m,\varepsilon}^{\lambda\pm}[\xi] =(H-\lambda)\phi_{m,\varepsilon}^{\lambda\pm}[\xi],
 \end{align*}  
 and  note the following formulas for the generalized eigenfunctions
 of \eqref{eq:rep1}, cf. \eqref{eq:basyp},
  \begin{align}\label{eq:rep12i3}
    \begin{split}
    \phi_{\rm{ex}}^{\lambda\pm}[\xi]= \phi_{m,\varepsilon}^{\lambda\pm}[\xi]-R(\lambda\mp\i 0)\psi_{m,\varepsilon}^{\lambda\pm}[\xi];\quad\xi\in
    C_\c^\infty(\R^{d-1}).  
    \end{split}
  \end{align}  
Since $\vF^{-}(\lambda)^*
  \xi=\phi_{\rm{ex}}^{\lambda-}[\xi]$ (cf. \eqref{eq:exForm}),  Theorem  
  \ref{thm:char-gener-eigenf-1} then leads to
  \begin{align*}
    T(\lambda)\xi:=S(\lambda)\xi-\xi= -\i 2\pi
  \vF^+(\lambda)\psi_{m,\varepsilon}^{\lambda-}[\xi].
  \end{align*}
   Using again 
  \eqref{eq:exForm} and 
  \eqref{eq:rep12i3} we conclude  the following
  analogue of \eqref{eq:formS2} and \eqref{eq:formS22i},
\begin{align}\label{eq:formS22i3}
  \begin{split}
    \tfrac 1{2\pi\i}\inp{{\xi},T(\lambda)\xi'}&=- \inp{{\vF^{+}(\lambda)^*\xi},\psi_{m,\varepsilon}^{\lambda-}[\xi']} \\&  =\inp{\psi_{m,\varepsilon}^{\lambda+}[{\xi}],R(\lambda+\i
      0)\psi_{m,\varepsilon}^{\lambda-}[\xi']}-\inp{\phi_{m,\varepsilon}^{\lambda+}[{\xi}],\psi_{m,\varepsilon}^{\lambda-}[\xi']}.
\end{split}
\end{align} 

\subsection{The leading order symbol of $T(\lambda)$} \label{subsubsec:Tbetter Best results on the scattering
  matrix}

The operator  $T(\lambda)$ is a pseudodifferential operator of {order}
$-\delta$, meaning that   its kernel can be written as
\begin{align*}
  T(\zeta,\zeta') = (2\pi)^{1-d}\int\, \e^{\i(\zeta-\zeta')\cdot
  y} \, t(\zeta,y)\,\d y,
\end{align*} where locally uniformly  in $\zeta$ 
(possibly locally uniformly   $\lambda$ as well) 
 \begin{align}\label{eq:depro3}
  \partial_{\zeta}^\alpha\,\partial_{y}^\beta
   t=\vO\parbb{\inp{y}^{-\delta -\abs{\beta}}}.
\end{align} 
 This definition  is consistent with \eqref{eq:orderDef} by the theory
 of PsDOs, cf. Step IV of the proof of Theorem \ref{thm:analys-scatt-matra}.

\begin{thm}\label{thm:ledSing} The operator  $T(\lambda)$ has order
  $-\delta$. The principal  symbol of
  $T(\lambda)$ is given by $t_{\rm psym}$
   in the sense that $T(\lambda)-T_{\rm psym}$ has order
  $-2\delta$; here  $T_{\rm psym}$ denotes the quantization of $ t_{\rm psym}=t_{\rm psym}(y):=-2\i
                                                                  \int_0^\infty\,
                                                                  \tfrac
                                                                  {q_1(x,-y)}{\sqrt{2x}}\,\d
  x$. 
  \end{thm}
  \begin{proof} \Step{I}
Due  to Step
I of the proof Corollary \ref{cor:impr-form-scatt} (including a
trivial modification of Lemma \ref{lemma:best-results-scatt})  and the
subsequent discussion the first term  on the right-hand side of \eqref{eq:formS22i3} is
represented by  a smoothing operator.

\Step{II}
    As for the second term we write 
\begin{align*}
     &-2\pi \i \inp{\phi_{m,\varepsilon}^{\lambda+}[{\xi}],\psi_{m,\varepsilon}^{\lambda-}[\xi']}=
\tfrac{-\i}{(2\pi)^{d}}\int\!\!\!\!\int\d x\d y \int  \,\d
  \zeta \,\overline{\xi(\zeta)}\\&\int \e^{-\i
    \theta_\lambda}\overline{a^{+}_B +\chi_1 \chi^{\perp+}_\varepsilon}\,\d \eta\int \,\d
  \zeta' \,\xi'(\zeta')\,\int \e^{\i
    \theta'_\lambda}\parb{\vO\parb{f^{-\infty}}+r^-_B  }\,\d \eta',
\end{align*}

The contribution from the term $\vO\parb{f^{-\infty}}$ is represented
by a smoothing operator,
since we can use the bound {$(x+\inp{y})^{-s}$,} $s$ large, and
$\eta$- and $\eta'$-integration by parts to  bound  any
derivative, 
\begin{align*}
  \partial^\alpha_{\zeta,\zeta',y} t^-(\zeta,\zeta',y)=\vO(\inp{(\zeta,\zeta',y)}^{-\infty})
\end{align*}  for the corresponding symbol $ t^-$, cf. Step
I of the proof  Theorem \ref{thm:analys-scatt-matra}.

\Step{III}  The `leading  order' contribution  is by \eqref{eq:leadT}  given
by 
\begin{align*}
     &-2\pi \i \inp{\phi_{m,\varepsilon}^{\lambda+}[{\xi}],\psi_{m,\varepsilon}^{\lambda-}[\xi']}\approx
\tfrac{-\i}{(2\pi)^{d}}\int\!\!\!\!\int\d x\d y \int  \,\d
  \zeta \,\overline{\xi(\zeta)}\int \e^{-\i
    \theta_\lambda}\overline{a^{+}_B +\chi_1 \chi^{\perp+}_\varepsilon}\d \eta\\&\int \,\d
  \zeta' \xi'(\zeta')\,\int \e^{\i\theta'_\lambda}\parb{-\i b_1^{-}\, \parb{\chi_1\partial_{\eta'}\chi^-_\varepsilon
+(\eta',\zeta') \cdot
\nabla_{(x,y)}(\chi_1\chi^-_\varepsilon)}+q \chi_1\chi^{\perp-}_\varepsilon}\d \eta'.
   \end{align*} 
 
We may treat
the contribution from the error from  the above approximation by mimicking 
 Steps II--IV  of the proof  Theorem \ref{thm:analys-scatt-matra}.  It
 is 
   represented by a 
     PsDO of order $-2\delta$, as wanted. Note that indeed the same cut-off functions as in Step
     II in the proof of Theorem \ref{thm:analys-scatt-matra}   apply.

     However there is a somewhat alternative treatment based on the van der
     Corput lemma, cf. \cite[p. 332]{St}, and more related to
     Subsection \ref{subsec:The stationary phase method}. Since we
     need it below, this method is explained here. We consider the
     PsDO
\begin{align*}
  \breve R(\zeta,\zeta') &=(2\pi)^{1-d}\int\, \e^{\i(\zeta-\zeta')\cdot
  y} \,\breve r(\zeta,\zeta',-y) \,\d y;\\
\breve r(\zeta,\zeta',y)&=( 2\pi \i)^{-1}\int\d
                   x\int {\e^{-\i\varphi_\lambda(x, \eta,\zeta)}\,\overline{a^{+}_B +\chi_1 \chi^{\perp+}_\varepsilon}}\,\d
                   \eta\int{\e^{\i\varphi_\lambda(x, \eta',\zeta') }} \breve r^-_B\,\d
  \eta';\\&\quad \varphi_\lambda(x, \eta,\zeta)={-\eta^3/6+(x+\lambda-\zeta^2/2)\eta}.
\end{align*} Here $\breve r^-_B$ is the exact  error from
\eqref{eq:leadT}. Proceeding as  in Step
     II of the proof Theorem \ref{thm:analys-scatt-matra} we can
     freely (i.e. up to a  term representing a  smoothing operator) insert the factors
     \begin{align*}
      \chi_\epsilon(\cdot):=\chi_\epsilon\parb{{\eta'}-\sqrt{2x}}\mand  \chi_{-\varepsilon}(\cdot):=\chi_{-\varepsilon}\parbb{-\tfrac{\eta'}{\sqrt{2x+2\inp{y}}}}
     \end{align*}
in the 
$\eta'$-integral, 
$\chi_\epsilon\parb{{\eta}-\sqrt{2x}}$ in the $\eta$-integral and
finally $\chi_R(x)$ and  $\chi(x,y;\epsilon)$  
 in the $x$-integral. As before $R>2$ is large and
$\epsilon>0$ is small. Next we introduce   $h^{-1}=\sqrt{2x}$  and  write
 \begin{align*}
  \varphi_\lambda\parb{=\varphi_\lambda(x,
   \eta,\zeta)}=h^{-1}\,\tilde{\varphi}_\lambda\,\mand \,\varphi'_\lambda\parb{=\varphi_\lambda(x,
   \eta',\zeta')}=h^{-1}\,\tilde{\varphi}'_\lambda.
\end{align*}  To
treat the $\eta'$-integral  we rewrite it as
\begin{align*}
 \int{\e^{\i\varphi_\lambda(x, \eta',\zeta') }} \breve r^-_B\chi_\epsilon(\cdot)\chi_{-\varepsilon}(\cdot)\,\d
  \eta'=\int \parbb{\tfrac{\d}{\d
  \eta'}\int^{\eta'}_{\sqrt{2x}}{\e^{\i\varphi_\lambda(x, s,\zeta') }}
  \d s}\breve r^-_B \chi_\epsilon(\cdot)\chi_{-\varepsilon}(\cdot)\,\d
  \eta', 
\end{align*} and integrate by parts. By the van der
     Corput lemma
     \begin{align*}
     \Big|\int^{\eta'}_{\sqrt{2x}}{\e^{\i h^{-1}\,\tilde{\varphi}'_\lambda}
       \d s} \Big| \leq C\sqrt h
     \end{align*} with a universal constant $C$. We proceed similarly for the
     $\eta$-integral. With the bounds on  $\breve r^-_B$  and $a_B^+$
     from
\eqref{eq:leadT} we then obtain, thanks to the factor
$\chi(x,y;\epsilon)$ and \eqref{eq:elebnd},
\begin{align}\label{eq:rBrev}
  \abs{\breve r(\zeta,\zeta',y)}\leq C_1 \int\chi_R(x) {f^{-(4\delta
     +1)}}x^{-1/2}\,\d
                   x \leq C_2\inp{y}^{-2\delta}
\end{align} with locally bounded constants (in
$(\zeta,\zeta')$). Derivatives can be treated as in Step IV  of the
proof  Theorem \ref{thm:analys-scatt-matra}. So in conclusion, 
 $\breve R$ is a PsDO of order $-2\delta$.

\Step{IV} We show  that the `leading  order' contribution in Step III
given by the symbol $ r^-_B-\breve r^-_B$  is represented by 
     a PsDO of order $-\delta$.

First we compute, cf. \eqref{eq:solut},
   \begin{align*}
     \breve r^-_B-r^-_B=-\i b_1^{-}\, \parb{\chi_1\partial_{\eta'}&\chi^-_\varepsilon
+(\eta',\zeta') \cdot
\nabla_{(x,y)}(\chi_1\chi^-_\varepsilon)}+q \chi_1\chi^{\perp-}_\varepsilon\\&= q \chi_1-\i\parb{\partial_{\eta'}
+(\eta',\zeta') \cdot
\nabla_{(x,y)}}\parb{b_1^{-}\chi_1\chi^-_\varepsilon}.
   \end{align*}  
The second term  contributes only  by a term of  order $-2\delta$, which may be  
seen as follows. By an $\eta'$-integration by part using
\eqref{eq:dereta}  we may  effectively replace 
\begin{align*}
  -\i\e^{\i\theta'_\lambda}\parb{\partial_{\eta'}
+(\eta',\zeta') \cdot
\nabla_{(x,y)}}\parb{b_1^{-}\chi_1\chi^-_\varepsilon}  \approx (H_0-\lambda)\parb{\e^{\i\theta'_\lambda}b_1^{-}\chi_1\chi^-_\varepsilon}
\end{align*}  with  error
\begin{align*}
  \tfrac12 \e^{\i\theta'_\lambda}p^2\parb{b_1^{-}\chi_1\chi^-_\varepsilon}=\vO\parbb{f^{-{2\delta}}\min\parb{f^2,{\inp{y}}}^{-2}},
\end{align*} cf. 
\eqref{eq:leadT},   and for this error term we can proceed as in Step III. Now we take
the factor $ (H_0-\lambda)$ to the left  in the integral (by
integration by parts) thereby producing the factor $(H_0-\lambda)\parb{\e^{\i
    \theta_\lambda}(a^{+}_B +\chi_1 \chi^{\perp+}_\varepsilon)}$. By the same argument, now for the
$\eta$-integration, we have 
\begin{align*}
  (H_0-\lambda)\parb{\e^{\i\theta_\lambda}(a^{+}_B +\chi_1
  \chi^{\perp+}_\varepsilon)}\approx -q\e^{\i\theta_\lambda}a^{+}_B+ \parbb{1+\tfrac{\abs{\eta}}f}\vO\parbb{f^{-{2\delta}}\min\parb{f^2,{\inp{y}}}^{-1/2}},
\end{align*}
 and we can   proceed  as in Step III for the second term obtaining again a term of
 order $-2\delta$. However the first term
 $-q\e^{\i\theta_\lambda}a^{+}_B$ is missing a support property (used
 for the second one), preventing us from using the factor
 $\chi(x,y;\epsilon)$. Whence we    need  to consider the localization  factors
 $\chi_\epsilon\parb{-\eta-\sqrt{2x}}$  and
$\chi_\epsilon\parb{-\eta'-\sqrt{2x}}$, as well as 
$\chi_\epsilon\parb{\eta-\sqrt{2x}}$ and
$\chi_\epsilon\parb{\eta'-\sqrt{2x}}$ used  before. Since
${b_1^{-}\chi_1\chi^-_\varepsilon}=\vO(f^{-2\delta})$ and
$qa^{+}_B=\vO(f^{-1-2\delta})$ we may bound the  symbol corresponding
to the term  $-q\e^{\i\theta_\lambda}a^{+}_B$, with the localizations
factors in place, 
 exactly as in \eqref{eq:rBrev}.

For
the first term $q \chi_1$ we can invoke the first identity of
\eqref{eq:leadT} and  replace the factor $\overline{a^{+}_B
  +\chi_1 \chi^{\perp+}_\varepsilon}$  by $1$ with an error corresponding to 
 a PsDO of order $-2\delta$. This is justified as we treated  $-q\e^{\i\theta_\lambda}a^{+}_B$
 above. Whence we end up with the  `Born term'
\begin{align*}
  \breve T(\zeta,\zeta') &=(2\pi)^{1-d}\int\, \e^{\i(\zeta-\zeta')\cdot
  y} \,\breve t(\zeta,\zeta',-y) \,\d y;\\
  \breve t(\zeta,\zeta',y)&=( 2\pi \i)^{-1}\sum_{\pm}\int\d
                   x \,\chi_R(x)q(x,y)\\&\int \e^{-\i\varphi_\lambda(x, \eta,\zeta)}\,\chi_\epsilon\parb{\pm{\eta}-\sqrt{2x}}\,\d
                   \eta\int{\e^{\i\varphi_\lambda(x, \eta',\zeta') }} \chi_\epsilon\parb{\pm{\eta'}-\sqrt{2x}}\,\d
  \eta'.
\end{align*} The $\eta$- and $\eta'$-integrals essentially represent  Airy
functions and their asymptotics can be obtained by the stationary
phase method, for example by  a simplified
version of  Appendix \ref{Appendix}.  Alternatively we may invoke
\cite[(7.6.21)]{Ho}.  Anyhow we conclude that    $\breve T$ (and
therefore also 
$T(\lambda)$)  is a PsDO of
order $-\delta$ with  
leading order symbol    given in terms of  any big $R>1$ (making
$\eta=\sqrt{2x+2\lambda-\zeta^2}$ well-defined) by 
\begin{align*}
   \breve t_{\rm KN}(\zeta,y)\approx -4\i
                                                                  \int_R^\infty\,
                                                                  \tfrac
                                                                  {q(x,-y)}{\sqrt{2x+2\lambda-\zeta^2}}\sin^2(\cdot)\,\d
  x\approx -2\i
                                                                  \int_0^\infty\,
                                                                  \tfrac
                                                                  {q_1(x,-y)}{\sqrt{2x}}\,\d
  x=t_{\rm psym}(y). 
\end{align*} The quantization of the error is  of
order $-2\delta$.
\end{proof}

By the definition  \eqref{eq:orderDef} and according to Theorem
\ref{thm:ledSing} the corresponding symbol  $t-t_{\rm psym}$ fulfills
bounds (with derivatives) that are 
 locally uniform in
  $\zeta$. We remark that from the above proof one easily deduce that  theses bounds
  are
  also (in fact simultaneously) locally uniform in  $\lambda$. The operator $T_{\rm
    psym}$ has  order
  $-\delta$ and in general no better, see Subsection \ref{subsubsec:Tbetter Best results on the scattering
  matrix2}.

\subsection{Analysis of  the kernel of $T(\lambda)$ at the diagonal} \label{subsubsec:Tbetter Best results on the scattering
  matrix2}
We compute the top  order  singularity
  of  the kernel of $T(\lambda)$ at the diagonal for a class of slowly decaying
  potentials. This  is  done by combining Theorem 
  \ref{thm:ledSing} with  \cite[Lemma
  4.1]{IK}.  

For a homogeneous potential $q\approx \kappa r^{-\alpha}$,
$1/2<\alpha<d-1/2$, we can compute using \eqref{eq:elebnd}
\begin{align*}
  \int_0^\infty\,
                              \tfrac
                                                                  {q(x,-y)}{\sqrt{2x}}\,\d
  x\approx \kappa c_1\abs{y}^{\tfrac 12 -\alpha}\text{ for }\abs{y}\to \infty,
\end{align*} with
\begin{align*}
  c_1=2^{-3/2}\int _0^\infty (t+1)^{-\alpha/2}t ^{-3/4}\,\d t=2^{-3/2} \Gamma(1/4)\Gamma(\alpha/2-1/4)/\Gamma(\alpha/2),
\end{align*} cf. \cite[ (13.2.5), (13.5.10)]{AS}.
The Fourier transform of  $\abs{y}^{\tfrac 12 -\alpha}$ is known.
Whence, cf. \cite[Lemma
  4.1]{IK},  in this case the top  order singularity of the kernel of
$T(\lambda)$ is given by 
\begin{align*}
  &T(\zeta,\zeta')\approx \kappa c_2\abs{\zeta-\zeta}^{\tfrac 12
  +\alpha-d};\\
  c_2=&c_1(2\pi)^{1-d}(-2\i)(2\pi)^{(d-1)/2}2^{d/2-\alpha}\Gamma((d-1/2-\alpha)/2))/\Gamma((\alpha-1/2)/2))\\
=&-\i(2\pi)^{(1-d)/2}2^{(d-1)/2-\alpha}\Gamma(1/4)\Gamma(d/2-1/4-\alpha/2)/\Gamma(\alpha/2)
\end{align*}

 For the Coulomb potential  $q= \kappa r^{-1}$ with $d\geq 3$  the order of
 $T_{\rm psym}$ is   $-1/2$, and the
singularity (at the diagonal) is  the form
\begin{align}\label{eq:upper3}
  T(\zeta,\zeta')\approx\kappa c_2\abs{\zeta-\zeta'}^{3/2-d}.
\end{align}   The error is in this case of order
$\vO\parb{\abs{\zeta-\zeta'}^{2-d}}$, cf. \cite[Lemma
  4.1]{IK}.  Whence we can  
summarize
as follows.
\begin{corollary}\label{thm:Cou3} For $q=\kappa r^{-1}$ and  $d\geq 3$ the
  kernel 
  \begin{align*}
  S(\lambda)(\zeta,
  \zeta')-\delta(\zeta,
  \zeta')=  \kappa c_2\abs{\zeta-\zeta'}^{3/2-d}+ \vO\parb{\abs{\zeta-\zeta'}^{2-d}}
  \end{align*}
at the diagonal, locally uniformly in
$\zeta$, $\zeta'$ and $\lambda$. 
  \end{corollary}

For  $q=\kappa r^{-1}$ and  $d= 3$, we can use that
\begin{align*}
  \Gamma(1/4)\Gamma(d/2-1/4-\alpha/2)=\sqrt 2\pi\quad\mand\quad
  \Gamma(\alpha/2)=\sqrt \pi,
\end{align*} yielding in that case $c_2=-\i(2\pi)^{-1/2}$ and $T(\zeta,\zeta')\approx -\i\kappa(2\pi)^{-1/2}\abs{\zeta-\zeta'}^{-3/2}$.
This result agrees with \cite {KK1}, where 
  the exact asymptotics at the diagonal is derived by
  a different method relying on explicit calculations of integrals
  involving powers of the potential and the free Stark resolvent kernel. This method seems 
  restricted to homogeneous potentials. In our  approach (which is  valid for a
  wider class of potentials) the singularities  are `sitting' in an explicit
oscillatory integral, which is amenable to analysis. In particular we
extracted the  top order singularity for  homogeneous potentials from
the  oscillatory integral.

For
a partial  decomposition of the kernel of $S(\lambda)$ for the Coulomb
potential, see
\cite {KK2}, although  this paper does not study   the singularity problem.

\appendix

\section{Proof of \eqref{eq:AsySTA}} 
\label{Appendix}

For convenience we only consider
$\phi^+_{\lambda,\tilde{a}}[\xi]$ and $\lambda=0$. Whence we need to consider the
asymptotics of 
\begin{align*}
  c\int \d
   \zeta \,\xi(\zeta)\int \e^{\i h^{-1}\tilde\theta}\,\tilde a\,\chi_\epsilon\parb{{\eta}-h^{-1}}\chi_\epsilon\parb{\big|\zeta-
  hy\big|}\,\d \eta.
\end{align*} of a symbol $\tilde{a}=\tilde{a}(x,y;\eta,\zeta)$
obeying \eqref{eq:1symb}. Recall that there is exactly one relevant 
stationary  point for $x>R$ and $\abs{y}< C\sqrt{2x}$,
cf. \eqref{eq:asypFixed},  say  denoted
$z(h,y)=(\eta^+,\zeta^+)$. This is  given by
\begin{align*}
  \eta^+= \sqrt{x+ \parb{x^2 - y^2}^{1/2}}\mand  \zeta^+=y/\eta^+.
\end{align*} Let similarly $z=(\eta,\zeta)$. To obtain the
asymptotics as $h\to 0$ uniformly in $y$ we write the phase
$\tilde\theta=h \theta$ as 
\begin{align}\label{eq:morseEqa}
  \tilde\theta=\tilde\theta_{|z=z(h,y)}+  \tfrac 12 \inp{\phi,A\phi},
\text{ where } A=A(h,y))=\nabla^2_z \tilde\theta_{|z=z(h,y)}
\end{align} and  $\phi$ is a diffeomorphism in $z$ from an open 
neighbourhood $U$ of $z(h,y) $ onto  an open neighbourhood $V$ of
$0$ with 
$\phi_{|z=z(h,y)}=0$ and derivative $\nabla_z\phi_{|z=z(h,y)}=I$.
The existence of such map $\phi$ follows as in the  proof of \cite[Lemma 4.2]{II}
 and the computation
 \begin{align*}
   \nabla^2_z\theta=-\eta \parb{I+\vO(\zeta/\eta)},\text{ yielding
   roughly }\nabla^2_z\tilde\theta\approx -I.
 \end{align*} Indeed we can  introduce  $\breve
 \theta(\breve z)=\tilde\theta(z)$ by substituting  $z=\breve z+z(h,y) $, write
 \begin{align*}
   \breve\theta(\breve z)=\tilde\theta_{|z=z(h,y)}+\tfrac 12\inp{\breve
   z,B(\breve z)\breve z};\quad
   B=B(\breve z)=2\int_0^1(1-\tau) \nabla^2_z
   \breve\theta\parb{\tau \breve z+z(h,y)}\,\d \tau,
 \end{align*}
and use the inverse function
 theorem to solve 
 \begin{align*}
\Phi(\Gamma):=\Gamma A^{-1}\Gamma=B
   \end{align*} for a unique real  symmetric $d\times d$ matrix
   $\Gamma=\Gamma(\breve z)=\Gamma(\breve z, h,y)$ near
   $A=A(h,y)$. Note that $\Phi(A)=A=B(0)$. Then $ \phi(z):=A^{-1}\Gamma(\breve z)\breve z$ works in 
\eqref{eq:morseEqa}, in fact with $U=z(h,y)+B_r(0)$ where $B_r(0)$ is
the open ball centered at $0$ with  radius $r>0$ being independent of
$(h,y)$ (seen conveniently by using  \cite[Lemma 1.18]{Sc}). Fix such $r$.  One easily checks that $\phi$ has bounded
derivatives with bounds being independent of $(h,y)$ (using the same 
property of $\Gamma$).

The inverse map $\psi:V\to U$ has  derivatives which
similarly are bounded uniformly in $(h,y)$ (seen inductively by the  Fa\`a di Bruno formula). 
 We change variable $z\to \phi$ and write, possibly at this
point taking 
$\epsilon>0$  smaller  and $R=R(\epsilon)>2$ larger,
\begin{align*}
  &\int \d
   \zeta \,\xi(\zeta)\int \e^{\i h^{-1}\tilde\theta}\,\tilde
  a\,\chi_\epsilon\parb{{\eta}-h^{-1}}\chi_\epsilon\parb{\big|\zeta-
    hy\big|}\,\d \eta\\&=\e^{\i \theta_{|z=z(h,y)}}\int_V \e^{\i h^{-1}2^{-1} \inp{\phi,A\phi}}\,f(\phi)\,\d
  \phi;\\& \quad \quad f(\phi)=\parbb{\xi(\zeta)\tilde
  a(\cdot)\,\chi_\epsilon{(\eta-h^{-1})}\chi_\epsilon\parb{\big|\zeta-
  hy\big|}}\parb{\psi(\phi)}\abs{\det(\psi')(\phi)},
\end{align*} and compute 
\begin{align*}
  \forall \alpha:&\quad \partial^\alpha_\phi f=\vO (\inp{x,y}^0),\\
\forall \alpha:&\quad \norm{\partial^\alpha_\phi f}_2=\vO (\inp{x,y}^0),
\end{align*} (Note for the latter bounds that  $ \int 1_V \d \phi <\infty$.)  By the
Plancherel theorem and \cite[Theorem 7.6.1]{Ho} 
\begin{align*}
  \int \e^{\i h^{-1}2^{-1} \inp{\phi,A\phi}}\,f(\phi)\,\d
  \phi=h^{d/2}\e^{\i\pi \,{\rm sgn}(A)/4}\abs{\det(A)}^{-1/2}\int_{\R^d} 
  \e^{-\i h2^{-1} \inp{\breve\zeta,A^{-1}\breve\zeta}}\hat f(\breve\zeta) \,\d \breve\zeta.
\end{align*} By the inversion formula 
\begin{align*}
  \int \hat
  f(\breve \zeta) \,\d \zeta=(2\pi)^{d/2}f(0).
\end{align*} On the other hand  by using the bound
\begin{align*}
  \abs{\e^{-\i h2^{-1} \inp{\breve\zeta,A^{-1} \breve\zeta}}-1}\leq h2^{-1}\abs{\inp{\breve\zeta,A^{-1}\breve\zeta}},
\end{align*} we can estimate for any integer $n> 2+d/2$
\begin{align*}
  &\Big|\int_{\R^d} \parb{\e^{-\i h2^{-1} \inp{\breve\zeta,A^{-1}\breve\zeta}}-1}\hat
  f(\breve \zeta) \,\d \breve\zeta\Big|\\& \leq
  h C_1\max_{\abs{\alpha}\leq n}\,\norm{\partial_z^\alpha f}_2 \\&
                                                                   \leq
                                                                   C_2h.
\end{align*}

 Finally, by  invoking  $ A=-I+\vO(h)$ (uniformly in $y$) and
 \eqref{eq:asypFixed}, the asymptotics \eqref{eq:AsySTA} follows. \qed

 \begin{remark}\label{remark:statproof-eqrefeq:asysta} We used above only
   the {zeroth} order Taylor expansion of the Gaussian function of
   $\breve\zeta$ at zero. If the symbol  $\tilde{a}$ vanishes to any order at
   the stationary point $(\eta^+,\zeta^+)$, then  higher order Taylor
   expansion yields  that the integral is $\vO(h^\infty)$ rather than
   $\vO(h^{d/2})$ as proved above.
   \end{remark}

\section{Borel construction for \eqref{eq:Bsym}} 
\label{AppendixB}

We consider $c^\pm:=\sum_0^\infty  \chi_k
b^\pm_k$, where  $\chi_k=\chi(f/C_k>1)$ needs to be determined. Fix 
$\epsilon>0$ (for example take $\epsilon=1$).
 Thanks to \eqref{eq:derBndq22i} we can  for any  $k\in \N_0$ find a
 sufficiently big $C_k\geq 2$ such that
 
\begin{align}\label{eq:B1}
  \begin{split}
  \big|\partial_{\eta,\zeta}^\alpha\,&\partial_{x}^\beta
  \,\partial_{y}^\gamma b^\pm_k\big|\leq
  2^{-k}f^{\epsilon-(2k\delta
  +\abs{\alpha}
+2\abs{\beta}+2\abs{\gamma})} \\&\text{ for
  }\abs{\alpha}+\abs{\beta}+\abs{\gamma}\leq k, \quad \pm
                                                           a>-\varepsilon\quad  \text{and}\quad 
                                                           f>C_k,  
  \end{split}
\end{align} and clearly we can take $C_0=2$ and assume that
$C_k>1+C_{k-1}$ for $k\geq 1$.

Note that for all $l\in \N_0$ there exists $C(l)>0$ such
for all $k\in \N_0$
\begin{align}\label{eq:B2}
  \big|\partial_{x}^\beta
  \,\partial_{y}^\gamma {\chi_k}\big|\leq C(l)
  f^{
-2\abs{\beta}} \, \min\parb{f^2,\inp{y}_m}^{-\abs{\gamma}}\text{ for
  }\abs{\beta}+\abs{\gamma}\leq l.
\end{align}

By combining \eqref{eq:B1} and \eqref{eq:B2}  with  the product rule we conclude that  for all $l\in \N_0$ there exists $\breve{C}(l)>0$ such
for all $k\geq l$
\begin{align*}
  \big|\partial_{\eta,\zeta}^\alpha\,&\partial_{x}^\beta
  \,\partial_{y}^\gamma \parb{\chi_kb^\pm_k}\big|\leq
  \breve{C}(l)2^{-k}f^{\epsilon-(2k\delta
  +\abs{\alpha}
+2{\beta})} \, \min\parb{f^2,\inp{y}_m}^{-\abs{\gamma}}\\&\text{ for
  }\abs{\alpha}+\abs{\beta}+\abs{\gamma}\leq l \text{ and for }\pm a>-\varepsilon.
\end{align*}

By summing up we conclude that  $ c^\pm $ are well-defined  smooth
functions in
the regions $\set{\pm a>-\varepsilon}$, respectively,   with
bounds 
\begin{align*}
  \big|\partial_{\eta,\zeta}^\alpha\,&\partial_{x}^\beta
  \,\partial_{y}^\gamma c^\pm\big|\leq
 C_{\alpha,\beta,\gamma}f^{-(\abs{\alpha}
+2{\beta})} \, \min\parb{f^2,\inp{y}_m}^{-\abs{\gamma}};\quad \pm a>-\varepsilon.
\end{align*}

Since $a_B^\pm=\chi^\pm_\varepsilon  c^\pm$,  these bounds and the
product rule  yields the  first bound of
\eqref{eq:Borel} (however being only locally  uniform
 in  $\zeta$). Note at this point that 
\begin{align}\label{eq:B3}
  \big|\partial_{\eta,\zeta}^\alpha\,&\partial_{x}^\beta
  \,\partial_{y}^\gamma \chi^\pm_\varepsilon \big|\leq
 C_{\alpha,\beta,\gamma}f^{-(\abs{\alpha}
+2{\beta})} \, \min\parb{f^2,\inp{y}_m}^{-\abs{\gamma}};\quad f>C_0.
\end{align}

 For  the second assertion of \eqref{eq:Borel} the term
 $\vO\parb{f^{-\infty}}$ is given explicitly as
 \begin{align*}
   \vO\parb{f^{-\infty}}= \sum^\infty_{0}\chi^\pm_\varepsilon(\chi_{k+1}-\chi_k)\parb{qb^\pm_{k}-\tfrac 12 (\Delta_{(x,y)} b^\pm_{k})}.
 \end{align*} By using \eqref{eq:B1}--\eqref{eq:B3} one easily
 checks that indeed the
 right-hand side  is bounded along with all derivatives    by any
 inverse power of $f$ with a bounding constant being locally  uniform
 in  $\zeta$. This  completes the proof of  \eqref{eq:Borel}.  The
 related bounds \eqref{eq:modbnd} easily follow too.

\end{document}